\documentclass[12pt]{article}
\usepackage{amsmath}
\usepackage{amsthm}
\usepackage{amssymb}
\usepackage{mathrsfs}
\usepackage[usenames]{color}
\usepackage[latin5]{inputenc}
\usepackage{cite}
\usepackage{pdfsync}
\usepackage{enumitem}
\setlist{parsep=0pt,itemindent=0pt}
\definecolor{webgreen}{rgb}{0,0.5,0}

\theoremstyle{definition}

\newtheorem{prop}{Proposition}
\newtheorem{thm}{Theorem}
\newtheorem{cor}{Corollary}
\newtheorem{rmk}{Remark}

\newtheorem{exa}{Example}
\numberwithin{equation}{section}
\numberwithin{thm}{section}
\numberwithin{lemma}{section}
\numberwithin{prop}{section}
\numberwithin{cor}{section}
\numberwithin{rmk}{section}
\numberwithin{defn}{section}
\numberwithin{exa}{section}

\newcommand{\gen}[1]{\partial_{#1}}

\newcommand{\pr}[1]{\rm pr^{(#1)}}
\newcommand{\curl}[1]{ \left\{#1\right\} }
\newcommand{\lie}{\mathfrak g}

\newcommand{\nil}{\mathsf {h}}
\newcommand{\schr}{\mathsf {sch}}
\newcommand{\Schr}{\mathsf {Sch}}

\newcommand{\dx}{\partial_x}

\newcommand{\dt}{\partial_t}

\newcommand{\eqg}{\mathsf{G_E}}
\newcommand{\LO}{\mathsf L}
\newcommand{\lr}[1]{\langle{#1}\rangle}

\definecolor{darkolivegreen}{rgb}{0.333333, 0.419608, 0.1843140}

\DeclareMathOperator{\Sl}{sl}

\DeclareMathOperator{\SL}{SL}
\DeclareMathOperator{\Heis}{H}

\DeclareMathOperator{\Or}{O}
\DeclareMathOperator{\Orr}{o}

\DeclareMathOperator{\So}{so}
\DeclareMathOperator{\SO}{SO}

\DeclareMathOperator{\sech}{sech}

\begin{document}
\pagenumbering{arabic}
\clearpage
\thispagestyle{empty}

\title{\Large Equivalence and Symmetries for Linear Parabolic Equations and Applications Revisited}

\author{
F.~G\"ung\"{o}r\\ \small
Department of Mathematics, Faculty of Science and Letters,\\ \small Istanbul Technical University, 34469 Istanbul, Turkey \thanks{e-mail: gungorf@itu.edu.tr}}

\date{}

\maketitle

\tableofcontents


\begin{abstract}
A systematic and unified approach to transformations and symmetries of general second order linear parabolic partial differential equations is presented. Equivalence group is used to derive the Appell type transformations, specifically  Mehler's kernel in any dimension.   The complete symmetry group classification is re-performed. A new criterion which is necessary and sufficient for reduction to the standard heat equation by point transformations is established. A similar criterion is also valid for the equations to have a four- or six-dimensional symmetry group (nontrivial symmetry groups). In this situation, the basis elements are listed in terms of coefficients.  A number of illustrative examples are given.  In particular, some applications from the recent literature are re-examined in our new approach.  Applications include a comparative discussion of heat kernels based on group-invariant solutions and the idea of connecting Lie symmetries and classical integral transforms introduced by Craddock and his coworkers. Multidimensional parabolic PDEs of heat and Schr\"{o}dinger type are also considered.
\end{abstract}

\section{Introduction}
The purpose of this paper is to present a systematic and unified approach to linear parabolic equations of the form
\begin{equation}\label{LPE}
u_t=a(x,t)u_{xx}+b(x,t)u_x+c(x,t)u, \quad a\ne 0,\quad x\in \mathbb{R}, \quad t>0,
\end{equation}
where $a$, $b$, $c$ are arbitrary smooth functions, from the point of view of local equivalence and symmetry properties. The coefficients $a$ and $b$ are called diffusion and drift functions. This class arises as a fundamental model in many different areas of mathematics and physics such as diffusion processes, stochastic (Markov) processes, Brownian motion,  probability theory, financial mathematics, population genetics, quantum chaos and others. The celebrated Kolmogorov ($c=0$) and Fokker-Planck (FP) ($a_{xx}-b_x+c=0$, namely in divergence or conservative form) equations are special cases.

A large body of literature exists on applications of Lie symmetry methods to construct analytical solutions and solve initial, terminal and boundary value problems for Eqs. \eqref{LPE} with different coefficients. What happens with these works is to compute symmetries from scratch for each coefficient.  A more general approach to computation of symmetries can be found in   \cite{Bluman1980, SastriDunn1985, CicognaVitali1989, ShtelenStogny1989, CicognaVitali1990,  Hill1992, Gaeta1994}. A criterion that guarantees existence of 4- and 6-dimensional symmetry groups for the Fokker-Planck equations was given in \cite{ShtelenStogny1989}.

On the other hand, a related question is to ask when an equation in the class \eqref{LPE} is equivalent to the standard heat  equation under the invertible point transformations. There have been some attempts \cite{Ricciardi1976, Bluman1980, Ibragimov2002, JohnpillaiMahomed2001} to answer this question. The last two papers rely on the notion of differential invariants or semi-invariants under the action of the full equivalence group or subgroups using infinitesimal techniques. In Ref. \cite{Ibragimov2002}, a second order semi-invariant (with $a,b,c$ defined by a change of sign)
$$K=\frac{1}{2}b^2a_x+(a_t+a a_{xx}-a_x^2)b+(aa_x-ab)b_x-ab_t-a^2b_{xx}+2a^2c_x$$
was derived and  commented that the condition $K=0$ could possibly guarantee the equivalence.
Later, it was shown in \cite{JohnpillaiMahomed2001} that fulfilment of a complicated singular invariant equation expressed in terms of $K, a$ and their derivatives up to order six, (which we don't reproduce here)   is both necessary and sufficient for the existence of a point transformation mapping \eqref{LPE} to the heat equation. A work extending point transformations to nonlocal ones was carried out in \cite{BlumanShtelen2004}.

The main motivation of the present paper is to lend a fresh perspective to these two fundamental issues: computation of symmetries for the general coefficients and also the identification of group structure of the symmetry  group  when nontrivial symmetries (4- and 6-dimensional) are allowed and establishing a new criterion of local equivalence to the heat equation using only global approach rather than infinitesimal one. The criterion is expressed in terms of a second order differential semi-invariant in case of spatially varying coefficients (admitting at least one additional symmetry other than homogeneity and linear superposition).

The complete symmetry classification of \eqref{LPE} was given by Lie himself in \cite{Lie1881} as part of a classification of all second order linear PDEs in two dimensions. The classification was redone  by Ovsiannikov  \cite{Ovsiannikov1982}. Here we do the same from a slightly different point of view.  Conservation laws and potential symmetries of \eqref{LPE} was investigated in \cite{PopovychKunzingerIvanova2008a}. A local equivalence problem for the class \eqref{LPE} under a contact transformation pseudo-group was addressed in \cite{Morozov2007}.

We organise this paper as follows. In Section \ref{Section2} we discuss equivalence groups for linear parabolic equations in the general form \eqref{LPE} on the real line and for those in $n$-dimensions in the potential form \eqref{Heatn}. We then use them to derive a necessary and sufficient criterion for the reducibility to the standard heat equations (zero potential). We provide a general transformation formula and illustrate with a number of examples. Some of them have already appeared in the literature and we re-examine them in our framework.  We also re-construct Appell type transformations and heat kernels for equations with three different potentials using equivalence transformations. Section \ref{Section3} is devoted to computation of the nontrivial symmetry algebras  and identification of their Lie-algebraic structure. For a comparison we apply our  results to  some special cases chosen from the literature. Specifically, we show how our approach can be effectively used to deal with any equations with at least a nontrivial symmetry in a unified way.  In Section \ref{Section4} we study group-invariant and fundamental solutions and give a discussion of initial and boundary-value problems.  Finally, Section \ref{Section5} gives a brief summary of the results discussed throughout.

\section{Equivalence group and transformation to the heat equation}\label{Section2}

\subsection{One-dimensional case}
As we are interested in transformations preserving only the differential form of the equation we present the following proposition as our main ingredient.

\begin{prop}\label{pro1}
The equivalence group of $\eqg$ of Eq.  \eqref{LPE} is given by
\begin{equation}\label{EqTr}
\tilde{t}=T(t),  \quad \tilde{x}=X(x,t), \quad \tilde{u}=\theta(x,t)u+\psi(x,t),
\end{equation}
where $T$, $X$ and $\theta$ are arbitrary smooth functions of their arguments and satisfy $\dot{T}, X_x, \theta \ne 0$. The new coefficients transform by
\begin{equation}\label{newcoeff}
\begin{split}
\tilde{a}&=\frac{X_x^2}{\dot{T}}a, \\
\tilde{b}&=\frac{X_x}{\dot{T}}\left[b-2a \frac{\theta_x}{\theta}+a\frac{X_{xx}}{X_x}-\frac{X_t}{X_x}\right],\\
\tilde{c}&=\frac{1}{\dot{T}}\left[c-b \frac{\theta_x}{\theta}+2a \left(\frac{\theta_x}{\theta}\right)^2-a\frac{\theta_{xx}}{\theta}+\frac{\theta_t}{\theta}\right],\\
\end{split}
\end{equation}
and $\psi$ satisfies the homogeneity preserving condition $\LO{(\psi/\theta)}=0$, where $\LO$ is the linear operator $\LO=\dt-a\dx^2-b\dx-c$. The  explicit form of this condition is
\begin{equation}
c \psi -b \psi \frac{\theta_x}{\theta}+2a \psi \left(\frac{\theta_x}{\theta}\right)^2+b \psi_x-2a \psi_x \frac{\theta_x}{\theta}-a \psi \frac{\theta_{xx}}{\theta}+a \psi_{xx}+\psi \frac{\theta_t}{\theta}-\psi_t=0.
\end{equation}
This condition identically holds for $\psi=0$.
\end{prop}
Throughout this article we shall assume that the class \eqref{LPE} has at least one   symmetry other than homogeneity in $u$ (the multiplication of $u$ by a nonzero constant) and superposition of solutions. Later we shall see that the presence of such a symmetry transforms our equation into one with coefficients depending on space variable only
\begin{equation}\label{mainLPE}
u_t=a(x)u_{xx}+b(x)u_x+c(x)u, \quad a\ne 0.
\end{equation}
Such an equation is sometimes called time-invariant or time autonomous.

For  convenience we introduce two functions (following the notation adopted in \cite{Hill1992})
\begin{equation}\label{IJ}
I(x)=\int^{x}\frac{d\hat{x}}{\sqrt{a(\hat{x})}},  \quad J(x)=\frac{1}{\sqrt{a}}[\frac{a'(x)}{2}-b(x)]=\frac{d}{dx}(\sqrt{a(x)})-\frac{b(x)}{\sqrt{a(x)}}.
\end{equation}
We can set $a$ to 1 and $b$ to zero by choosing $X$ and $\theta$ as
$$X=\varepsilon \sqrt{\dot{T}}\int\frac{dx}{\sqrt{a}}+\delta(t)=\varepsilon \sqrt{\dot{T}}I(x)+\delta(t),\quad \dot{T}>0$$
$$\frac{\theta_x}{\theta}=\frac{b}{2a}-\frac{1}{4}\frac{a'}{a}-\frac{1}{8}\frac{\ddot{T}}{\dot{T}}(I^2)'-
\frac{\dot{\delta}}{2\sqrt{\dot{T}}}I'=-\frac{1}{2}\frac{J}{\sqrt{a}}-\frac{1}{8}\frac{\ddot{T}}{\dot{T}}(I^2)'-
\frac{\dot{\delta}}{2\sqrt{\dot{T}}}I',$$
where $\varepsilon=\pm 1$, $\delta$ is an arbitrary function. Integrating the second equation above, $\theta$ has the form
\begin{equation}\label{b2zero}
\theta(x,t)=\nu(t)\exp\left[-\frac{1}{2}\int\frac{J}{\sqrt{a}}dx-\frac{1}{8}\frac{\ddot{T}}{\dot{T}}I^2-
\frac{\dot{\delta}}{2\sqrt{\dot{T}}}I\right],
\end{equation}
where $\nu(t)\ne 0$ is another arbitrary function. Under the special case $T=t$, $\delta=0$ the initial equation \eqref{mainLPE} will reduce to the canonical (or potential) form
\begin{equation}\label{potentialc}
\tilde{u}_{\tilde{t}}-\tilde{u}_{\tilde{x}\tilde{x}}
=\left(\frac{1}{2}\sqrt{a}J'(x)-\frac{1}{4}J^2(x)+c\right)\tilde{u}=\tilde{c}(\tilde{x})\tilde{u}.
\end{equation}
If $b=b(x,t)$ the transformed coefficient $\tilde{c}$ is given by
$$\tilde{c}=\frac{1}{2}\sqrt{a}J_x-\frac{1}{4}J^2-\frac{1}{2}\int\frac{b_t}{a}dx+c,$$
where $J$ is redefined by
\begin{equation}\label{jxt}
  J(x,t)=\frac{1}{\sqrt{a}}[\frac{a'(x)}{2}-b(x,t)].
\end{equation}
We call the function
\begin{equation}\label{inv}
K(x)=\frac{1}{2}\sqrt{a}J'(x)-\frac{1}{4}J^2(x)+c
\end{equation}
an invariant (more precisely a semi-invariant) of Eq. \eqref{mainLPE}. In case $a=a(x)$, $b=b(x,t)$, $c=c(x,t)$ we replace it by
$$K(x,t)=\frac{1}{2}\sqrt{a}J_x-\frac{1}{4}J^2+c,$$ where $J$ is given by \eqref{jxt}.

$K$ is not changed under the change of dependent variable only (leaving $x$ and $t$ unaltered). $a$ is also a semi-invariant. Indeed, if $\theta$ is eliminated between the relations
$$\tilde{a}=a, \quad \tilde{b}=b-2a \frac{\theta_x}{\theta}, \quad \tilde{c}=c-b\frac{\theta_x}{\theta}+2a\left(\frac{\theta_x}{\theta}\right)^2 -a \frac{\theta_{xx}}{\theta},$$
it turns out $\tilde{K}=K$. The semi-invariant $K$ will play a role in determining locally equivalent equations. One can see that two equations from the initial class can be transformed into each other by a change of the dependent variable $\tilde{u}=\theta(x,t)u$, $\theta\ne 0$ if and only if the function $K$ is the same for both equations. We note that  our semi-invariant $K$ differentiated once with respect to $x$ actually coincides with that of Ref. \cite{Ibragimov2002} modulo trivial invariance of $a$ and its derivatives when pure spatial dependence on $x$ of $a$ is present while $b$ and $c$ can depend on both $x$ and $t$.

We now turn to our canonical equation in a form with a more general potential
\begin{equation}\label{potential}
u_t-u_{xx}+V(x,t)u=0.
\end{equation}
We restrict the equivalence transformations to find the subgroup of $\eqg$ preserving the relations $\tilde{a}=a=1$, $\tilde{b}=b=0$ and obtain the following proposition.
\begin{prop}\label{pro2}
The equivalence group $\mathsf{\tilde{G}_E}$ of Eq. \eqref{potential} consists of the transformations
\begin{equation}\label{normalEQT}
\begin{split}
\tilde{t}=T(t),  \quad \tilde{x}=\sqrt{\dot{T}}x+\delta(t), \quad \tilde{u}=\nu(t)\exp\left[-\frac{1}{8}\frac{\ddot{T}}{\dot{T}}x^2-
\frac{\dot{\delta}}{2\sqrt{\dot{T}}}x\right]u,\\
\tilde{V}=\frac{1}{\dot{T}}\left[V+\frac{1}{8}\curl{T;t}x^2+\frac{1}{2}\left(\frac{\ddot{\delta}}{\sqrt{\dot{T}}}-
\frac{\ddot{T}\dot{\delta}}{\dot{T}^{3/2}}\right)x-\frac{\dot{\delta}^2}{4\dot{T}}-\frac{\ddot{T}}{4\dot{T}}-\frac{\dot{\nu}}{\nu}\right],
\end{split}
\end{equation}
where $\curl{T;t}$ is the Schwarzian derivative of $T$ with respect to $t$ defined by
$$\curl{T;t}=\frac{\dddot{T}}{\dot{T}}-\frac{3}{2}\Bigl(\frac{\ddot{T}}{\dot{T}}\Bigr)^2.$$
\end{prop}
Choosing $\tilde{V}=0$ we have the following.
\begin{cor}
The most general equation of the form \eqref{potential} which can be transformed to the standard heat equation should have the form
$${V}=q_2(t)x^2+q_1(t)x+q_0(t).$$
\end{cor}
This common fact will be our starting point  in search of a practical test for transformability to the heat equation.
In order to transform an equation of the form
\begin{equation}\label{HeatplusQuad}
u_t-u_{xx}+(q_2(t)x^2+q_1(t)x+q_0(t))u=0
\end{equation}
to $u_t-u_{xx}=0$
we simply use the freedom left in the transformations \eqref{normalEQT}. We fix the coefficients of the quadratic potential as
\begin{equation}\label{coeffconds}
-\frac{1}{8}\curl{T;t}=q_2(t),  \quad -\frac{1}{2}\left(\frac{\ddot{\delta}}{\sqrt{\dot{T}}}-
\frac{\ddot{T}\dot{\delta}}{\dot{T}^{3/2}}\right)=q_1(t),  \quad \frac{\dot{\delta}^2}{4\dot{T}}+\frac{\ddot{T}}{4\dot{T}}+\frac{\dot{\nu}}{\nu}=q_0(t).
\end{equation}
and solve for $T$, $\delta$ and $\nu$.
The first equation
\begin{equation}\label{Schwarz}
\curl{T;t}=-8q_2(t)
\end{equation}
is a Schwarzian equation for $T(t)$ which is reduced to the linear second order equation by the ratio transformation
$$T(t)=\frac{k_1 \Omega_1+ k_2 \Omega_2}{k_3 \Omega_1+ k_4 \Omega_2}, \quad k_1k_4-k_2k_3\ne 0,$$
where $\Omega_1, \Omega_2$ are two independent solutions of the linear equation
\begin{equation}
\ddot{\Omega}-4q_2 \Omega=0.
\end{equation}
Note that the same Schwarzian equation can be transformed to a Riccati equation of the form
$$\dot{\varrho}-\frac{1}{2}\varrho^2=-8q_2(t).$$
by the transformation $\varrho=\ddot{T}/\dot{T}$. The remaining equations can be expressed in terms of a new function $\omega$ defined by $\delta(t)=\sqrt{\dot{T}}\omega(t)$ in a more compact form. This implies  $\tilde{x}=\sqrt{\dot{T}}(x+\omega(t))$ in \eqref{pro2}. The second equation is transformed to
$$\ddot{\omega}+\frac{1}{2}\curl{T;t} \omega=-2 q_1(t),$$ which, by virtue of the first Eq. in \eqref{potential}, becomes \begin{equation}\label{omegaEq}
\ddot{\omega}-4q_2(t)\omega=-2 q_1(t).
\end{equation}
Using the relation $\delta(t)=\sqrt{\dot{T}}\omega(t)$ and integrating the third equation of \eqref{coeffconds} we find a formula for $\nu(t)$ in terms of $T$ and $\omega$ in the form
\begin{equation}\label{nu}
\nu(t)=\nu_0^{-1}\dot{T}^{-1/4}\exp\curl{-\frac{\ddot{T}}{8\dot{T}}\omega^2-\frac{1}{4}\int[4q_2 \omega^2+\dot{\omega}^2-4q_0]dt},
\end{equation}
where $\nu_0\ne 0$ is arbitrary integration constant.
We conclude that the functions $T$, $\omega$, thereby $\delta$ and $\nu$ are completely determined by fixing the coefficients of the quadratic potential. We can state the point transformations achieving the reduction to the heat equation $\tilde{u}_{\tilde{t}}-\tilde{u}_{\tilde{x}\tilde{x}}=0$ in the following form
\begin{equation}\label{ToHeatTr}
\begin{split}
&\tilde{t}=T(t),  \quad \tilde{x}=\sqrt{\dot{T}}(x+\omega(t)), \\
&u=\nu_0 \dot{T}^{1/4}\exp \curl{\frac{\ddot{T}}{8\dot{T}}x^2+\frac{1}{2}(\dot{\omega}+\frac{\ddot{T}}{2\dot{T}}\omega)x+\frac{\ddot{T}}{8\dot{T}}\omega^2
+\frac{1}{4}\int[4q_2 \omega^2+\dot{\omega}^2-4q_0]dt}\tilde{u}.
\end{split}
\end{equation}
In other words, if $\tilde{u}$ solves the heat equation then so  $u$ will solve Eq.  \eqref{HeatplusQuad}.
We note that this transformation depends on six arbitrary constants, three of which comes form the solution of the Schwarzian equation, two from   \eqref{omegaEq} and one from $\nu_0$.

In summary we have established a practical test which will ensure  local equivalence of \eqref{mainLPE} to the heat equation.
\begin{prop}
Any equation from the class \eqref{mainLPE} can be transformed to the heat equation if and only if the semi-invariant equals a quadratic polynomial in $I$
\begin{equation}\label{Riccati}
K(x)=\frac{1}{2}\sqrt{a}J'-\frac{1}{4}J^2+c=q_2 I^2+q_1 I+q_0,
\end{equation}
where $q_2$, $q_1$ and $q_0$ are some constants. Furthermore,  symmetry group of the initial equation can be conjugated to the one of the heat equation producing isomorphic symmetry groups.
\end{prop}

\begin{rmk}
Condition \eqref{Riccati} is a  Riccati equation  for $J(x)$ and can be integrated  in some special cases. Changing the independent variable to $I$, the Riccati equation can be expressed as
\begin{equation}\label{Riccati2}
\frac{1}{2}\frac{dJ}{dI}-\frac{1}{4}J^2+c=q_2 I^2+q_1 I+q_0.
\end{equation}
The transformation $J=-2\pi'(I)/\pi(I)$ reduces \eqref{Riccati2} to
a second order linear ODE for $\pi$, which is actually a parabolic cylinder equation.
Given $a(x)$ and $c(x)$, this allows us to find classes that can be transformed to the heat equation. Hence,   $b(x)$ is extracted from $J$ so that the obtained equation is reducible to the heat equation. In particular,  for the Fokker-Planck equation we have $a=p(x)$, $b=p'(x)+q(x)$, $c(x)=q'(x)$, $J(x)=-[(\sqrt{p})'+q/\sqrt{p}]$. Furthermore, for $p=1$ the above condition becomes
$$q'-\frac{1}{2}q^2=c_2x^2+c_1x+c_0,$$ where $c_2$, $c_1$, $c_0$ are constants. The class of Fokker-Planck equations with $p=1$, $q(x)$ arbitrary was investigated in \cite{BlumanCole1974}.
\end{rmk}

For an equation in the class \eqref{mainLPE} with the coefficients given, a simple strategy to decide about reducibility to the heat equation consists of computing $K(x)$ and inspecting if it can be written as a quadratic polynomial in $I$ (in general, time dependent  coefficients can be allowed). When this is the case the corresponding transformation is given by
\begin{equation}\label{ToHeatTr2}
\begin{split}
&\tilde{t}=T(t),  \quad \tilde{x}=\sqrt{\dot{T}}(I+\omega(t)), \\
&\tilde{u}=\nu(t)\exp\left[-\frac{1}{2}\int\frac{J}{\sqrt{a}}dx-\frac{1}{8}\frac{\ddot{T}}{\dot{T}}I^2-
\frac{1}{2}(\dot{\omega}+\frac{\ddot{T}}{2\dot{T}}\omega)I\right]u,
\end{split}
\end{equation}
where $T$, $\omega$ are solutions to \eqref{Schwarz}, \eqref{omegaEq} and $\nu$ is as before (formula \eqref{nu}).

As an illustration we  analyse the subclass where only the diffusion coefficient is present.
\begin{exa}
$$u_t=a(x)u_{xx}.$$
\end{exa}
In this case $b=c=0$, $J=a'/(2\sqrt{a})=(\sqrt{a})'$,
$$K=\frac{1}{2}A A''-\frac{1}{4}A'^2,  \quad A=\sqrt{a}.$$ We suppose that $K$ is equal to a constant so that
$$A A''-\frac{1}{2}A'^2=m.$$ We differentiate it to get
$A'''=0$ which has the general solution $A(x)=a_2 x^2+a_1 x+a_0$. We have found a polynomial diffusion $a(x)=(a_2 x^2+a_1 x+a_0)^2$ which is transformable to the heat equation.  It includes the powers $a(x)=\alpha x^4$, $a(x)=\beta x^2$.

The special case $a(x)=(1+k^2x^2)^2$, $k\ne 0$ arises in a study of Brownian motion. We have
$$I(x)=\frac{1}{k}\arctan (kx), \quad J(x)=2k^2x,  \quad K(x)=k^2.$$ From \eqref{Riccati} we have $q_2=q_1=0$, $q_0=-k^2$ so from \eqref{ToHeatTr2} the corresponding transformation is found to be
$$\tilde{t}=t,  \quad \tilde{x}=\frac{1}{k}\arctan (kx), \quad \tilde{u}=e^{-k^2t}(1+k^2x^2)^{-1/2}u.$$
This case was discussed in \cite{Winternitz1989} as an example of identifying isomorphic symmetry algebras of PDEs.

For the power diffusion $a(x)=\sigma x^{2\gamma}$, $\gamma\ne 0$ we have
$$I=\frac{x^{1-\gamma}}{\sqrt{\sigma}(1-\gamma)}, \quad J=\gamma  \sqrt{\sigma }x^{-(1-\gamma) }, \quad K=\frac{\gamma(\gamma-2)}{4(1-\gamma)^2}I^{-2}, \quad \gamma\ne 1.$$ From $K$ we see that reduction to the heat equation can be possible only for $\gamma=2$, otherwise it would be reduced to the second canonical form of the heat equation with a four-dimensional symmetry group (see \eqref{4-dim}). For $\gamma=1$ we have $I=\ln x/\sqrt{\sigma}$, $J=\sqrt{\sigma}$, $K=-\sigma/4$, which implies that reduction to the heat equation is possible. The relevant transformations can directly be constructed from \eqref{ToHeatTr2} by the choice of $q_2=q_1=0$, $q_0=\sigma/4$. For $a=x^2$ it has the form
$$\tilde{t}=t,  \quad \tilde{x}=\ln x, \quad u=\sqrt{x}e^{-t/4}\tilde{u}.$$
In summary, there are only two powers where the equations are equivalent to the heat equation.

It is straightforward to check that  for the FP equation
$$u_t=\frac{\partial}{\partial x}[(Ax+B)^{2\gamma}u_x],\quad A\ne 0,$$ it follows $K=0, -A^2/4$ for $\gamma\in\curl{2/3,1}$, respectively. Both is equivalent to the heat equation. The transformation for the first case is
$$\tilde{t}=t,  \quad \tilde{x}=\frac{3}{A}(Ax+B)^{1/3},  \quad u=(Ax+B)^{-1/3}\tilde{u}.$$ For the other case we put $q_2=q_1=0$, $q_0=A^2/4$ and find
$$\tilde{t}=t,  \quad \tilde{x}=\frac{1}{A}\ln (Ax+B),  \quad u=e^{-A^2/(4t)}(Ax+B)^{-1/2}\tilde{u}.$$

\begin{exa}[Ref. \cite{SpichakStognii1999}]
$$u_t=\frac{\partial^2}{\partial x^2}[(1-x^2)^2u].$$
We have $a=(1-x^2)^2$, $b=-8x(1-x^2)$, $c=4(3x^2-1)$ and $J=6x$, $K=-1$, which indicates that the equation is reducible to the heat equation. Using formula \eqref{ToHeatTr2} for $q_2=q_1=0$, $q_0=1$ we find the transformation to be
$$\tilde{x}=\tanh^{-1} x=\frac{1}{2}\ln\frac{1+x}{1-x},  \quad \tilde{t}=t,  \quad u=(1-x^2)^{-3/2}e^{-t}\tilde{u}.$$
\end{exa}

\begin{exa}\label{ex-arb-drift}
$$u_t=u_{xx}+f(x)u_x.$$
We have $a=1$, $b=f(x)$, $c=0$, $I=x$, $J=-f(x)$, $K=-\frac{1}{2}f'-\frac{1}{4}f^2.$ Let us check the condition (Riccati equation)
\begin{equation}\label{Riccati3}
K=-\frac{1}{2}(f'+\frac{1}{2}f^2)=-(q_2 x^2+q_1 x+q_0)
\end{equation}
for some possible drift function $f$. For $f=2/x$, we have $K=0$ which indicates equivalence to the heat equation via $u=\tilde{u}/x$ (from \eqref{b2zero}). For $f=k/x$, $k\ne 0, 2$, we have $K=-[{k(k-2)}/{4}]x^{-2}$ so the equation would be in the other canonical class via the transformation $u=x^{-k/2}\tilde{u}$. Of course, other special solutions of the Riccati equation can produce examples of equations transformable to the canonical forms. For example, if $f(x)=bx$, $b\ne 0$ (see also \cite{BoykoShapoval2011}) then \eqref{Riccati3} is satisfied for $q_2=b^2/4$, $q_1=0$, $q_0=b/2$ so that reduction to heat equation can be achieved by a transformation (formula \eqref{ToHeatTr2})
$$\tilde{t}=\frac{1}{2b}e^{2bt},  \quad \tilde{x}=e^{bt}x, \quad \tilde{u}=u.$$ Note that this transformation is not unique.
Also, for $f(x)=\tanh(x/2+c)$ with $c$ a constant the left hand side of \eqref{Riccati3} is equal to a constant
$$f'+\frac{1}{2}f^2=\frac{1}{2},$$ from which it follows that $q_2=q_1=0$, $q_0=-1/4$. Again the corresponding equation is equivalent to the heat equation via the map
$$\tilde{t}=t,  \quad \tilde{x}=x,  \quad u=e^{-t/4}\sech \left(\frac{x}{2}+c\right)\;\tilde{u}.$$

In general, the  transformation $f=2w'/w$ takes \eqref{Riccati3} to the second order linear equation
$$w''-h(x)w=0, \quad h(x)=q_2 x^2+q_1 x+q_0,$$
known as parabolic cylinder equation which belongs to the class of generalized hypergeometric type equations \footnote{A generalized hypergeometric equation \cite{NikiforovUvarov1988} is one of the form $$y''+\frac{p(x)}{q(x)}y'+\frac{r(x)}{q^2(x)}y=0,$$ where $p$ is a linear polynomial, $q$ and $r$ are quadratic polynomials at most. A linear change of $y$ can be used to transform it the usual hypergeometric equation where $r=\lambda q$, $\lambda$ a constant.}   which can be reduced to a Hermite equation by a  change of dependent variable or to a confluent hypergeometric equation by a further change of independent variable. It is always possible to eliminate $q_1$ by a translation in $x$. In the special cases  $w''-q_2 x^2 w=0$ and $w''-(q_1 x+q_0) w=0$  solutions are expressed in terms of Bessel functions of index $1/4$ and $1/3$, respectively. The latter equation is called Airy equation.

Recall that the radial heat equation in $n$ dimensions satisfies
\begin{equation}\label{radial-heat}
  u_t=u_{xx}+\frac{k}{x}u_x, \quad k=n-1,
\end{equation}
where $x$ is the radial variable $x=|x|$, $0<x<\infty$.
\end{exa}

The transformations \eqref{ToHeatTr} can also be used to derive the famous Appell transformations (discrete symmetries), heat kernels (fundamental solutions) for the heat equation, the harmonic oscillator or Hermite heat equation and the heat equation with linear potential.
\begin{exa}
The heat equation $u_t-u_{xx}=0$.

We set $q_2=q_1=q_0=0$ and pick functions $T(t)$, $\omega(t)$  as solutions of the equations
$$\curl{T;t}=0,  \quad \ddot{\omega}=0.$$
The general solution is
$$T(t)=\frac{at+b}{ct+d},\quad \Delta=ad-bc>0, \quad \omega(t)=\omega_1 t+\omega_0,$$
where $a$, $b$, $c$, $d$, $\omega_1$ and $\omega_0$ are constants. Note that $\dot{T}=\Delta (ct+d)^{-2}$, $\ddot{T}/\dot{T}=-2c/(ct+d)$. We find a formula involving 6 arbitrary constants  taking one solution of the heat equation to another one (a discrete symmetry). For the special choice $\tilde{t}=T=-1/t$ ($\dot{T}=t^{-2}$), $\omega=-y=\text{const.}$ we find the translated Appell transformation
\begin{equation}\label{Appell-heat}
  u=c_0 t^{-1/2}\exp\left[-\frac{(x-y)^2}{4t}\right]\tilde{u}(\tilde{t},\tilde{x}), \quad T=-\frac{1}{t}, \quad \tilde{x}=\frac{x-y}{t}.
\end{equation}
This result was originally derived in \cite{Appell1982}.
Choosing the constant solution $\tilde{u}=1$ and $c_0=(4\pi)^{-1/2}$ we recover the fundamental solution
\begin{equation}\label{fundsol}
  K(x,t,y)=(4\pi t)^{-1/2}\exp\left[-\frac{(x-y)^2}{4t}\right],  \quad t>0
\end{equation}
with singularity at $(y,0)$ for the heat equation with the initial condition
$$\lim_{t\to 0^{+}}K(x,t,y)=\delta(x-y),$$ where $\delta$ is the Dirac distribution and the limit is to be taken in the  distributional sense. Since the equation is invariant under $x$-translations we can translate $K(x,t,0)$ to get $K(x,t,y)=K(x-y,t,0)$. Later on, this simple idea of using translation group will be applied to variable coefficient equations with nontrivial symmetries.  The choice of the normalizing constant $c_0$ is dictated by the  condition
$$\lim_{t\to 0^+}\int_{\mathbb{R}}K(x,t,y)dx=1.$$
Observe that
$$\lim_{t\to 0}[-4t \ln K(x,t,0)]=|x|^2,$$ which is the squared euclidean distance. Such asymptotic behaviors of heat kernels for variable coefficient heat equations was undertaken by Varadhan \cite{Varadhan1967}.
\end{exa}

\begin{exa}
The heat equation with linear potential: $$u_t-u_{xx}-x u=0.$$
We have $q_2=q_0=0$, $q_0=1$, $\ddot{\omega}=-2$. This implies that $T(t)$ is a fractional linear (or M\"{o}bius) transformation of $t$ and $\omega(t)=\omega_1 t+\omega_0-t^2$. We take $T(t)=-1/t$, $\omega(t)=-t^2-y$ and find
the transformation rule
$$u=c_0t^{-1/2}\exp\left[-\frac{(x-y)^2}{4t}+\frac{t^3}{12}-\frac{t}{2}(x+y)\right]\tilde{u}(\tilde{t},\tilde{x}), \quad T=-\frac{1}{t}, \quad \tilde{x}=\frac{1}{t}(x-t^2-y).$$
We can now choose $\tilde{u}=1$ and $c_0=(4\pi)^{-1/2}$ and obtain the fundamental solution
$$K(x,t,y)=\frac{1}{\sqrt{4\pi t}}\exp\left[-\frac{(x-y)^2}{4t}+\frac{t^3}{12}-\frac{t}{2}(x+y)\right].$$
\end{exa}

\begin{exa}
The harmonic oscillator equation:
\begin{equation}\label{harmonic-oscillator-epsilon}
  u_t-u_{xx}+\varepsilon x^2 u=0, \quad \varepsilon=\pm 1.
\end{equation}
For $\varepsilon=1$, we put $q_2=1$, $q_1=q_0=0$. $T$ and $\omega$ satisfy
$$\curl{T,t}=-8,  \quad \ddot{\omega}-4\omega=0.$$ We choose the special solutions $T(t)=-1/2\coth(2t)$,  $\omega(t)=-y\cosh(2t)$ and obtain
\begin{equation}
\begin{split}
&u=c_0(\sinh 2t)^{-1/2}\exp\left[-\frac{\cosh 2t(x^2+y^2)-2xy}{2\sinh 2t}\right] \tilde{u}(\tilde{t},\tilde{x}),\\
&\tilde{t}=-\frac{1}{2}\coth 2t, \quad \tilde{x}=(\sinh 2t)^{-1}(x-y \cosh 2t).
\end{split}
\end{equation}
For the special choice $c_0=1/\sqrt{2\pi}$, $\tilde{u}=1$ we have obtained the Mehler's formula for the heat kernel of the harmonic oscillator equation without using Mehler's Hermite polynomial formula
$$K(x,t,y)= (2\pi\sinh 2t)^{-1/2}\exp\left[-\frac{\cosh 2t(x^2+y^2)-2xy}{2\sinh 2t}\right].$$

For $\varepsilon=-1$ we replace the hyperbolic functions figuring in the above formula by the trigonometric functions to find the corresponding Mehler's formula (or Mehler kernel)
$$K(x,t,y)= (2\pi\sin 2t)^{-1/2}\exp\left[-\frac{\cos 2t(x^2+y^2)-2xy}{2\sin 2t}\right].$$

We note that a slightly different form of the equation
$$u_t-u_{xx}+\lambda^2 x^2 u=0, \quad \lambda\ne 0$$ can be scaled to \eqref{harmonic-oscillator-epsilon} by a scaling of the independent variables: $(t,x)\to (\lambda t, \sqrt{\lambda}x)$.
\end{exa}
The construction of heat kernels by Lie symmetry groups as an alternative to other different methods in the literature will be presented in Section \eqref{Section4}.

\subsection{Heat equation with arbitrary potential in $n$-dimensions}\label{Sectionheatn}
Consider
\begin{equation}\label{Heatn}
u_t=\Delta u+V(x,t)u,  \quad (x,t)\in \mathbb{R}^n\times (0,\infty),
\end{equation}
where $\Delta$ is the usual Laplacian on $\mathbb{R}^n$. We want to construct the equivalence group of this equation. We look for invertible transformations $\tilde{x}_{\mu}=X_{\mu}(x, u)$, $ \tilde{u}=U(x,u)$ which preserves the form of the equation \eqref{Heatn}. For ease of notational simplicity we put $x=(t,x)=(x_0,x_1,\ldots, x_n)$.
We need first and second order derivational relations between old and new coordinates
\begin{equation}\label{first-order-der}
(U_u-X_{\sigma u}\tilde{u}_{\tilde{\sigma}})u_{\mu}
=\tilde{u}_{\tilde{\sigma}}X_{\sigma, \mu}-U_{\mu}
\end{equation}
and
\begin{equation}\label{second-order-der}
[U_u-\tilde{u}_{\tilde{\sigma}}X_{\sigma u}]u_{\mu\nu}=\tilde{u}_{\tilde{\sigma}\tilde{\rho}}(D_{\mu}X_{\rho})(D_{\nu}X_{\sigma}) + \tilde{u}_{\tilde{\sigma}}V_{\mu\nu}\cdot X_{\sigma}-V_{\mu\nu}\cdot U,
\end{equation}
where $D$ is the total differentiation operator, $\displaystyle \tilde{u}_{\tilde{\sigma}}=\partial \tilde{u}/\partial \tilde{x}^{\sigma}$, and we have defined
$$V_{\mu\nu}\cdot U=U_{\mu\nu} + u_{\mu}U_{\nu u} + u_{\nu}U_{\mu u} + u_{\mu}u_{\nu}U_{uu}$$ and a similar expression for $V_{\mu\nu}\cdot X_{\sigma}$ and we sum over repeated indices. We form $u_0-\delta_{\mu\nu}u_{\mu\nu}=u_0-u_{\mu\mu}=u_0-\Delta u$  from \eqref{first-order-der} and \eqref{second-order-der} and obtain
\begin{equation}\label{transformedeqn}
\begin{split}
   & [U_u-\tilde{u}_{\tilde{\sigma}}X_{\sigma u}](u_0-\Delta u)= [X_{0,0}\tilde{u}_0-(D_{\mu}X_{\rho})(D_{\mu}X_{\rho} )\Delta \tilde{u}] \\
    & -\tilde{u}_{\tilde{\sigma}\tilde{\rho}}(D_{\mu}X_{\rho})(D_{\mu}X_{\sigma}) + \tilde{u}_{\tilde{\sigma}}(X_{\sigma,0}-V\cdot X_{\sigma})+V\cdot U-U_0, \quad \rho\ne \sigma,
\end{split}
\end{equation}
where $V\cdot U=\Delta U + 2u_{\nu}U_{\nu u} + u_{\nu}u_{\nu}U_{uu}$ and a similar expression for $V\cdot X_{\sigma}$. From the coefficients of $\tilde{u}_{0\rho}$, $\tilde{u}_{\sigma \rho}$, $\sigma\ne \rho$ we have
$$X_{0,\rho}=0, \quad X_{\mu u}=0, \quad  (D_{\mu}X_{\rho})(D_{\mu}X_{\sigma})=0, \quad \mu=0,1,\ldots n, \quad \sigma,\rho=1,\ldots n.$$
We also require the first term on the right of \eqref{transformedeqn} to be proportional to the heat operator so that we can write
\begin{equation}\label{mixed-der}
  X_{\sigma,\mu}X_{\rho,\mu}=\lambda^2(x_0)\delta_{\sigma \rho}, \quad \lambda^2(x_0)=X_{0,0}> 0, \quad \sigma,\rho=1,2,\ldots, n.
\end{equation}
This relation indicates that $X_a$ should be linear in $x_k$:
$$X_k=\lambda(x_0)A_{kl}x_l+\beta_k(x_0), \quad k,l=1,2,\ldots n, \quad k\ne l,$$
where $A_{kl}\in \Or(n)$, $\Or(n)$ being the group of $n\times n$ orthogonal matrices and $\beta_k(x_0)$ is arbitrary. Eq. \eqref{transformedeqn} has the form
$$U_u\mathcal{H} u=\lambda(t)\tilde{\mathcal{H}}\tilde{u} + \tilde{u}_{\tilde{\sigma}}(X_{\sigma,0}-\Delta X_{\sigma}) + (\Delta U -U_0 + 2u_{\nu}U_{\nu u} + u_{\nu}u_{\nu}U_{uu}), \quad U_u\ne 0,$$ where $\mathcal{H}=\partial_0-\Delta $ is the heat operator. $\Delta X_{\sigma}=0$ because $X_k$ is linear in $x_k$. The terms linear and quadratic in the first  derivatives must vanish.
So, after substituting $\displaystyle u_{\mu}= \frac{\tilde{u}_{\tilde{\sigma}}X_{\sigma,\mu}}{U_u}- \frac{U_{\mu}}{U_u}$  we find
$$2X_{\sigma,\nu}\frac{R_{\nu}}{R}+X_{\sigma,0}=0,\quad U=R(x)u+S(x), \quad R\ne 0.$$ For convenience we put $R_{\nu}/R=-F_{\nu}$ (or $R=e^{-F}$) in the first relation
$$F_{\nu}=\frac{1}{2\lambda}(\dot{\lambda}x_{\nu}+\dot{\beta}_{\mu}A_{\mu\nu}),$$ which integrates to
\begin{equation}\label{F}
  F(t,x)=\frac{\dot{\lambda}}{\lambda}\frac{|x|^2}{4}+ \frac{\dot{\beta}_k}{2\lambda} A_{kl}x_l+f(t),
\end{equation}
where $|x|$ is the euclidean norm of $x$ and $f(t)$ is an arbitrary function of integration.
From now on we set $x_0=t$, $x=(t,x)=(t,x_1,\ldots x_n)$ and $X_0=T(t)$ ($\lambda(t)=\sqrt{\dot{T}}$).
In order to find the transformed potential, we replace $\mathcal{H} u$ and $\tilde{\mathcal{H}} \tilde{u}$ by $V(t,x)u$ and $\tilde{V}(\tilde{t},\tilde{x})\tilde{u}$. This gives, using the relation $-U_0+\Delta U=U(F_t-\Delta F+|\nabla F|^2)$ and  simplifying further,
\begin{equation}\label{transformed-pot-n}
  \dot{T}\tilde{V}=V+\Delta F-|\nabla F|^2-F_t.
\end{equation}
or in terms of $R$
\begin{equation}\label{transformed-pot-n-2}
  (\dot{T}\tilde{V}-V)R=R_t-\Delta R.
\end{equation}
On  using the following computations
$$\Delta F=\frac{n}{2}\frac{\dot{\lambda}}{\lambda},  \quad |\nabla F|^2=\left(\frac{\dot{\lambda}}{\lambda}\right)^2\frac{|x|^2}{4}+\frac{\dot{\lambda}}{2\lambda^2} \dot{\beta}_k A_{kl} x_l+\frac{1}{4\lambda^2}\sum_{k=1}^n {\dot{\beta}_k}^2,$$
$$ F_t=\frac{d}{dt}\left(\frac{\dot{\lambda}}{\lambda}\right)\frac{|x|^2}{4}+ \frac{d}{dt}\left(\frac{\dot{\beta}_k}{2\lambda}\right)A_{kl} x_l+\dot{f}(t),$$
we can express $\tilde{V}$ in the form
\begin{subequations}\label{potential-rel-n}
\begin{equation}
\tilde{V}=\frac{1}{\dot{T}}\left[V-A(t)|x|^2-\sum_{k=1}^n B_k(t)A_{kl} x_l-\sum_{k=1}^n C_k(t)\right],
\end{equation}
where
\begin{equation}
  A(t)=\frac{1}{8}\curl{T;t},  \quad B_k(t)=-\frac{\dot{\lambda}\dot{\beta}_k}{2\lambda^2}+\frac{d}{dt}\frac{\dot{\beta}_k}{2\lambda}, \quad C_k(t)=\dot{f}-\frac{n}{4}\frac{\ddot{T}}{\dot{T}}-\frac{1}{4}\left(\frac{\dot{\beta}_k}{\lambda}\right)^2.
\end{equation}
\end{subequations}
Note that $S$ has to be a solution of the original equation.
So we have the following result.
\begin{prop}\label{pro3}
The equivalence group $\mathsf{\tilde{G}_E}$ of \eqref{Heatn} consists of fiber-preserving transformations
\begin{equation}\label{Equivn}
\tilde{t}=T(t), \quad \tilde{x}_k=X_k=\sqrt{\dot{T}} A_{kl}x_l+\beta_k(t), \quad u(x,t)=\exp[F(x,t)] \tilde{u}(\tilde{x},\tilde{t}),
\end{equation}
where $A_{kl}A_{km}=\delta_{lm}$ and $F$ is given by \eqref{F}.
 The new and old potentials are related by \eqref{potential-rel-n}.
\end{prop}
Here we have ignored the term $S(t,x)$ which gives rise to the linear superposition principle.
As a by-product we have demonstrated that the most general potential $V$ that can be  transformed to one with zero potential (the standard heat equation) should be of the form
\begin{equation}\label{potentialV}
  V(x,t)=a(t)|x|^2+\sum_{k=1}^n b_k(t) x_k+\sum_{k=1}^n c_k(t)
\end{equation}
for some functions $a$, $b_k$ and $c_k$ of $t$.

Just as in one dimensional case, we can use the equivalence group
\begin{equation}\label{Equivn2}
\begin{split}
\tilde{t}&=T(t), \quad \tilde{x}_k= \sqrt{\dot{T}(t)}A_{kl}x_l+\beta_k(t), \quad u(x,t)=\exp[F(x,t)] \tilde{u}(\tilde{x},\tilde{t}),\\
F&=\frac{{\ddot{T}}}{\dot{T}}\frac{|x|^2}{8}+ \frac{\dot{\beta}_k}{2\sqrt{\dot{T}}} A_{kl} x_l+f(t)
\end{split}
\end{equation}
to find the Appell type transformations or fundamental solutions. Again we find it convenient to change $\beta_k$ to $\omega_k$ via $\beta_k=\sqrt{\dot{T}} \omega_k$. The functions $T$, $\omega_k$ and $f$ are found from solutions of
$$\curl{T;t}=8a(t),  \quad \ddot{\omega}_k+4a\omega_k=2b_k(t),$$
$$\dot{f}=\frac{n}{4}\frac{\ddot{T}}{\dot{T}}+\sum_{k=1}^n\left[\frac{1}{4}\frac{\ddot{T}}{\dot{T}}\omega_k\dot{\omega}_k+\frac{1}{4}\dot{\omega}_k^2+
\frac{1}{16}\left(\frac{\ddot{T}}{\dot{T}}\right)^2\omega_k^2+c_k(t)\right].$$

\begin{exa}
Linear potential case:
$$u_t=\Delta u+\left(\sum_{k=1}^n b_k x_k\right)u,  \quad b_k=\rm{const.}$$

We put $a(t)=0$, $b_k(t)=b_k$, $c_k(t)=0$ and $T(t)=-1/t$, $\omega_k(t)=-b_k t^2-y_k$. From the last equation above we find
$$f(t)=f_0 -\frac{n}{2}\ln t+\sum_{k=1}^n\left[b_k^2\frac{t^3}{12}-\frac{t}{2}b_k y_k-\frac{y_k^2}{4t}\right],$$
which leads to the Mehler's formula \cite{CraddockLennox2012} up to a multiplicative constant
$$K(x,t,y)=c_0 t^{-n/2}\exp\left[-\frac{|x-y|^2}{4t}+\frac{t^3}{12}\sum_{k=1}^n b_k^2 -\frac{t}{2}\sum_{k=1}^n b_k (x_k+y_k)\right].$$
\end{exa}
In particular, by choosing $b_k=0$, $k=1,2,\ldots ,n$ in the above formula we obtain the fundamental solution on the $n$-dimensional euclidean space $\mathbb{R}^n$ for the standard heat equation
$$K(x,t,y)=c_0 t^{-n/2}\exp\left[-\frac{|x-y|^2}{4t}\right],$$ up to a multiplicative constant. $c_0$ can be specified from the limit
\begin{equation}\label{distrib-limit}
  \lim_{\varepsilon\to 0^{^{+}}}\frac{e^{-\frac{|x|^2}{\varepsilon}}}{(\pi \varepsilon)^{n/2}}=\delta(x), \quad x\in \mathbb{R}^{{n}}
\end{equation}
to be $c_0=(4\pi)^{-n/2}$.

The derivation of the fundamental solution
$$K(x,t,y)=c_0 (\sin 2t)^{-n/2}\exp\left[-\frac{\cos 2t (|x|^2+|y|^2)-2x\cdot y}{2\sin 2t}\right]$$
for the harmonic heat equation
$$u_t=\Delta u+|x|^2u$$ is left to the reader. Of course, it is also possible to derive the corresponding Appel transformation, (which first appeared in \cite{Goff1927} in $n$-dimensions) taking solutions among themselves.

The situation for the most general variable coefficient case becomes much more complicated in higher dimensions $n\geq 2$. We treat the following somewhat general  case  in two space-dimensions ($n=2$) where only a spatial dependence on the coefficients (assumed to be smooth) is present
\begin{equation}\label{genvcparabolic}
  u_t=a(x,y)u_{xx}+b(x,y)u_{xy}+c(x,y)u_{yy}+d(x,y)u_x+e(x,y)u_y+f(x,y,t)u,
\end{equation}
where $b^2-4ac<0$ in some subset of $\mathbb{R}^2$.
First of all, we perform a standard canonical (locally invertible) transformation $(x,y)\to (X(x,y),Y(x,y))$ on the spatial variables only in such a way that the new coefficients satisfy $\tilde{a}=\tilde{c}$ and $\tilde{b}=0$ or more explicitly
$$\sigma(X_x,X_y)=\sigma(Y_x,Y_y), \qquad aX_xY_x+\frac{b}{2}(X_xY_y+X_yY_x)+cX_yY_y=0,$$ where
$\sigma(r_1,r_2)=ar_1^2+br_1r_2+cr_2^2$. If it happens that $\sigma(X_x,X_y)=K=\text{const.}$, the principal part $L_p u=a(x,y)u_{xx}+b(x,y)u_{xy}+c(x,y)u_{yy}$ is then reduced to a $K$ multiple of the Laplacian $\Delta U=U_{XX}+U_{YY}$, while the remaining terms of \eqref{genvcparabolic} remain form invariant. The factor $K$ can be scaled to 1 by scaling time.   This is of course always true if the coefficients $a$, $b$, $c$ are all constants. For example, the nonconstant triple $(a,b,c)=(x^2,0,y^2)$ in $Q=\mathbb{R}^{+}\times \mathbb{R}^{+}$ can be transformed to $(1,0,1)$ by the transformation $X=\log x$, $Y=\log y$ with  the remaining coefficients changed (the lower order derivative term $-(U_X+U_Y)$ added). In general,  $L_p u=A x^2u_{xx}+B x yu_{xy}+Cy^2u_{yy}$, $B^2-4AC<0$, $(x,y)\in Q$ is transformable to the Laplacian in an analogues way by some coordinate transformation. We shall suppose that equation \eqref{genvcparabolic} has  already been reduced to
$$u_t=u_{xx}+u_{yy}+p(x,y)u_x+q(x,y)u_y+V(x,y,t)u.$$
Under the equivalence transformation
$$\tilde{t}=\tau(t),  \quad \tilde{x}=\sqrt{\dot{\tau}}(x+\delta(t)), \quad \tilde{y}=\sqrt{\dot{\tau}}(y+\rho(t)), \quad u=e^{F(x,y,t)}U(\tilde{x},\tilde{y},\tilde{t}),$$ the above equation is transformed to
$$U_{\tilde{t}}=U_{\tilde{x}\tilde{x}}+U_{\tilde{y}\tilde{y}}+\tilde{p}U_{\tilde{x}}
+\tilde{q}U_{\tilde{y}}+\tilde{V}U,$$ where
$$\tilde{p}=2F_x+p-\frac{\ddot{\tau}}{2\dot{\tau}}(x+\delta)-\dot{\delta},\quad \tilde{q}=2F_y+q-\frac{\ddot{\tau}}{2\dot{\tau}}(y+\rho)-\dot{\rho}, \quad \tilde{V}=\frac{1}{\dot{\tau}}[V-W]$$ with
$$W=F_t-(F_{xx}+F_{yy})-(F_x^2+F_y^2)-pF_x-qF_y.$$ We can make $\tilde{p}=\tilde{q}=0$ by a suitable choice of $F$ iff $p_y=q_x$. Under this assumption $F$ is constructed as
$$F=\frac{\ddot{\tau}}{\dot{\tau}}\left[\frac{1}{8}(x^2+y^2)+\frac{1}{4}(\delta x+\rho y)\right]+\frac{1}{2}(\dot{\delta}x+\dot{\rho}y)-\frac{1}{2}\phi(x,y)+\log \nu(t),$$ where $\nu$ is arbitrary integration function and
$d\phi(x,y)=p(x,y)dx+q(x,y)dy$. With this choice of $F$  we can express $W$ in the form
\begin{equation}\label{W}
  \begin{split}
     W(x,y,t)&=\frac{1}{8}\curl{\tau;t}(x^2+y^2)+\Omega(\tau,\delta,p)x+\Omega(\tau,\rho,q)y+\frac{1}{2}(p_x+q_y)-\frac{1}{4}(p^2+q^2)\\ &+\frac{1}{2}(p \dot{\delta}+q \dot{\rho})+\frac{\dot{\nu}}{\nu}-\frac{1}{4}(\dot{\delta}^2+\dot{\rho}^2)-\frac{\ddot{\tau}}{2\dot{\tau}}
     \left[1-\frac{1}{2}(p \delta+q \rho)+\frac{1}{4}\frac{d}{dt}(\delta^2+\rho^2)\right]\\
     & -\frac{1}{16}\left(\frac{\ddot{\tau}}{\dot{\tau}}\right)^2(\delta^2+\rho^2),\\
     \Omega(\tau,s_1,s_2)&\equiv\frac{s_1}{4}\curl{\tau;t}+\frac{\ddot{s}_1}{2}+\frac{\ddot{\tau}}{4\dot{\tau}}s_2.
   \end{split}
\end{equation}
Putting $p=q=0$ in $W$ and choosing $\tilde{V}=0$ we deduce that an equation of the form  \eqref{genvcparabolic} can be mapped to the usual heat equation  if and only if
$$V(x,y,t)=A(t)(x^2+y^2)+B_1(t)x+B_2(t)y+C(t)$$  for some given functions $A$, $B_1$, $B_2$ and $C$ which can annulled by an appropriate choice of the free functions $\tau$, $\delta$, $\rho$ and $\nu$ (Compare with \eqref{potentialV}).

As an example, one can map the PDE $u_t=u_{xx}+u_{yy}-u_x-u_y$ to its potential form with $\tilde{V}=-1/2$, which is then reduced to the heat equation in $U$  with zero potential in $2+1$ dimensions by the transformation $u=e^{-(t-x-y)/2}U(x,y,t)$. Hence, $u_t=x^2u_{xx}+y^2u_{yy}$ is reduced to the same canonical form by the transformation $$u=e^{-\frac{t}{2}}\sqrt{xy}\,U(X,Y,t), \quad X=\log x, \quad Y=\log y.$$

\section{Symmetry classification of the general parabolic equation}\label{Section3}
We know that \eqref{LPE} is invariant under scalings generated by $M=u\gen u$ and the infinite dimensional symmetry group generated by $\mathbf{v}_{f}=f(x,t)\gen u$, where $f$ solves the equation. They characterize  linearity of the equation. All other infinitesimal symmetries will be generated by vector fields of the form
\begin{equation}\label{VF}
X=\tau(t)\gen t+\xi(x,t)\gen x+\phi(x,t)u\gen u,
\end{equation}
where the coefficients $\tau, \xi, \phi$ will be determined from the Lie's symmetry condition (determining equations obtained from the requirement that the second prolongation $\pr 2 X$ annihilates Eq. \eqref{LPE} on the solution set)
\begin{equation}\label{deteqs}
\begin{split}
&8 a^2 \phi _x + \left(a_x - 2 b\right) \left(\tau  a_t+\xi  a_x-a \dot{\tau}\right) - a [2 \xi  a_{xx}+2 \tau(a_{xt}-2 b_t) -4 \xi _t-4 \xi  b_x ]=0,\\
&c \tau _t+\tau  c_t-\phi _t+\xi  c_x+b \phi _x+a \phi _{xx}=0,
\end{split}
\end{equation}
where $\xi$ satisfies the first order PDE
\begin{equation}\label{xieq}
\xi_x -\frac{a_x}{2a}\xi=\frac{1}{2a}(a\tau)_t.
\end{equation}
The $u$-coefficient $\phi$ will be determined up to a constant. We already know the equivalence group of the equation. We wish to extend the trivial symmetries by vector fields of the form \eqref{VF}. $X$ and $M$ commute, $[X, M]=0$. This implies that $\curl{X ,M}$ forms a two-dimensional abelian algebra. The operator $M$ is invariant under the equivalence transformations \eqref{EqTr}, whereas $X$ gets transformed to
$$\tilde{X}=\dot{T}\tau\gen {\tilde{t}}+(\xi X_x+\tau X_t)\gen {\tilde{x}}+[(\tau \theta_t+\xi \theta_x+\psi\theta )u+\tau \psi_t+\xi \psi_x]\gen {\tilde{u}}.$$
\begin{enumerate}
\item Let $\tau. \xi\ne 0$. Restricting \eqref{EqTr} to the  solutions (at least local) of the system
$$\dot{T}\tau=1,  \quad  \xi X_x+\tau X_t=0,  \quad \tau \theta_t+\xi \theta_x+\psi \theta =0,  \quad \tau \psi_t+\xi \psi_x=0$$ reduces $X$ to $\tilde{X}=\gen {\tilde{t}}$.

\item Let $\tau\ne 0$, $\xi=0$. We choose $\dot{T}\tau=1$, $\psi=\psi(x)$ and $\theta$ a solution of
$$ \tau \theta_t+\psi \theta =0, \quad \theta\ne 0$$ so that $X$ has been reduced to $\tilde{X}=\gen {\tilde{t}}$ again.

\item If $\tau=0$, $\xi\ne 0$, We choose $X$ to be $\xi X_x=1$, $\psi=\psi(t)$ and $\theta$ as a solution of
$$\xi \theta_x+\psi \theta =0, \quad \theta\ne 0.$$ This reduces $X$ to the canonical form $\tilde{X}=\gen {\tilde{x}}$.

If $\tau. \xi= 0$, then we have $\phi(x,t)\ne 0$.
\end{enumerate}
So we have found three different realizations of the two dimensional abelian symmetry algebra
$$  \lie_1=\langle\gen t, u\gen u\rangle, \quad \lie_2=\langle\gen x, u\gen u\rangle,  \quad \lie_3=\langle \phi(x,t)u\gen t, u\gen u\rangle,  \quad \phi(x,t)\ne \text{const.}$$
$\lie_3$ is not admissible as a symmetry algebra. We look at $\lie_2$ which leads to the invariant equation
$$u_t=a(t) u_{xx}+b(t)u_x+c(t) u.$$ We can reparametrise time to scale out $a$. Then we use the point transformation
$$\tilde{t}=t, \quad \tilde{x}=x+\int b(t)dt, \quad \tilde{u}=\exp[-\int c(t)dt]u$$
to transform $b$ and $c$ away simultaneously which results in the standard heat equation for which the symmetry group is well-known.

Now we turn to $\lie_1$. The corresponding invariant equation becomes \eqref{mainLPE}. The following transformations leave $\lie_1$ invariant:
$$\tilde{t}=t+t_0,  \quad \tilde{x}=X(x),  \quad \tilde{u}=\theta(x)u+\psi(x).$$
$X$, $\theta$, $\psi$ can be chosen appropriately so that \eqref{mainLPE} can be transformed to the canonical form
\begin{equation}
u_t=u_{xx}+V(x)u.
\end{equation}
Putting $a=1$, $b=0$ and $c=V(x)$ in \eqref{deteqs} and solving  we find
$$\xi(x,t)=\frac{1}{2}\dot{\tau}x+\eta(t),  \quad \phi(x,t)=-\frac{1}{8}\ddot{\tau}x^2-\frac{1}{2}\dot{\eta}x+h(t)$$
and $V(x)$ satisfies the classifying equation
\begin{equation}\label{ODE}
4(\dot{\tau}x+2\eta(t))V'(x)+8\dot{\tau}V(x)=-\dddot{\tau}x^2-4\ddot{\eta}x+8\dot{h}.
\end{equation}
Differentiating this equation three times gives
\begin{equation}\label{ODE2}
2\eta V^{(4)}+\dot{\tau}(5 V^{(3)}+x V^{(4)})=0.
\end{equation}
There are three cases:
\begin{enumerate}
\item $\eta= 0$, $\dot{\tau}=0$.  From \eqref{ODE} $h(t)=h_0=\textrm{const.}$ There is now new symmetry.

\item $\eta\ne 0$, $\dot{\tau}=0$. Up to an equivalence transformation $\eta\to 0$. We find again $\lie_1$ algebra.

\item $\eta= 0$, $\dot{\tau}\ne 0$. Then we obtain $V(x)$ from \eqref{ODE2} as
\begin{equation}\label{PE}
V(x)=\mu x^{-2}+c_2 x^2+c_1 x+c_0,
\end{equation}
where $\mu, c_2, c_1, c_0$ are constants.  If $\mu=0$ we can set $c_2=c_1=c_0=0$ ($V(x)=0$) by the equivalence transformations \eqref{ToHeatTr}. This gives us the heat equation $u_t=u_{xx}$. From \eqref{ODE} we have
$$\dddot{\tau}=0,  \quad \ddot{\eta}=0,  \quad \dot{h}=0$$ which leads to the six-dimensional symmetry algebra of the heat equation
\begin{equation}\label{Heat-Symm}
  \lie^{(6)}=\lr{ T, D, C, P, B, M   },
\end{equation}
where
\begin{equation}\label{Heat-Symm-comm}
  \begin{split}
     &T=\gen t, \quad D= t\gen t+\frac{x}{2}\gen x, \quad C= t^2\gen t+xt\gen x-\frac{1}{4}(x^2+2t)u\gen u, \\
      & P=\gen x, \quad B=t\gen x-\frac{xu}{2}\gen u, \quad M= u\gen u.
   \end{split}
\end{equation}
The commutation relations \eqref{comm-6-dim-c20} show that $\lr{T, D, C}$ is the Lie algebra $\Sl(2,\mathbb{R})$ and $\lr{P, B, M}$ the Heisenberg algebra $\nil(3)$. The Lie symmetry algebra has the semi-direct sum structure
$$\lie^{(6)}=\Sl(2,\mathbb{R})\vartriangleright \nil(3),$$
where $\nil(3)$ is the radical (maximal solvable ideal) of the algebra. We recall that the heat equation like any linear equation is also invariant under the infinite-dimensional abelian ideal $\lie_{\infty}=\lr{f(x,t)\gen u}$, where $f$ is a solution of the heat equation. The maximal symmetry algebra has the structure
$$\lie =\lie^{(6)}\vartriangleright \lie_{\infty}.$$

Now we let $\mu\ne 0$. From Eq. \eqref{ODE} it follows that $\eta=0$ and $c_1 \dot{\tau}=0$. If $\dot{\tau}=0$ then $\dot{h}=0$ which gives trivial symmetry algebra. For $c_1=0$ there are nontrivial symmetries which form a four-dimensional symmetry group. Using the equivalence transformations
$$\tilde{t}=T(t),  \quad \tilde{x}=\sqrt{\dot T}x, \quad u={\dot T}^{1/4}\exp[\frac{1}{8}x^2+c_0 t]\tilde{u},$$ where $T$ is a solution of $\curl{T;t}=8c_2$, for example, for $c_2>0$ one can choose $T(t)=\tan (2\sqrt{c_2}t)$,
we can put $c_2=c_0=0$ and obtain the canonical equation
\begin{equation}\label{canonical2}
  u_t-u_{xx}=\frac{\mu}{x^2}u,  \quad \mu\ne 0.
\end{equation}
A basis for the symmetry algebra is obtained by solving
$$\dddot{\tau}=0,  \quad \ddot{\tau}+4 \dot{h}=0$$ in the form
\begin{equation}\label{canonical-2-Lie}
  \lie^{(4)}=\lr{T, D, C, M}\sim \Sl(2,\mathbb{R})\vartriangleright\lr{M},
\end{equation}
where
\begin{equation}\label{canonical-2-basis}
   T=\gen t, \quad D=t\gen t+\frac{x}{2}\gen x, \quad C=t^2\gen t+xt\gen x-\frac{1}{4}(x^2+2t)u\gen u, \quad M=u\gen u.
\end{equation}

\item $\eta\ne 0$, $\dot{\tau}\ne 0$. This case is equivalent to the case 3. up to equivalence. We can write Eq. \eqref{ODE2} in the form
\begin{equation}\label{V(x)-eq}
  \frac{V^{(4)}}{V^{(3)}}=-\frac{5}{x+k},  \quad k=\frac{2\eta}{\dot{\tau}}.
\end{equation}
As $k$ should be  a constant it can be set to zero by translations of $x$.   The general solution of \eqref{V(x)-eq} with $k=0$ is again \eqref{PE}.
\end{enumerate}
In summary we have found two canonical forms
\begin{equation}\label{canonical}
u_t-u_{xx}=0,  \quad u_t-u_{xx}=\frac{\mu}{x^2}u, \quad \mu\ne 0,
\end{equation}
the first of which has a symmetry group
$G=\SL(2,\mathbb{R})\ltimes \Heis(1)$, and the second one $G=\SL(2,\mathbb{R})\ltimes \mathbb{R}$, where $\ltimes$ denotes the semi-direct product, $\Heis(1)$ is the 3-dimensional Hesisenberg algebra.

\subsection{Lie symmetries of \eqref{mainLPE}}
In what follows we assume that the equation under study is forward-propagating evolution $a>0$ for $t>t_0$ with the initial condition $u(x,t_0)$ given. From symmetry point of view, forward and backward ($t<t_0$)  type evolutions of diffusion processes differ only by the transformation $t\to t_0-t$.
This fact implies that they are irreversible in the sense that forward time (future) is distinguishable from backward time (past).

We now would like to solve the determining equations \eqref{deteqs} in the case $a=a(x)$, $b=b(x)$, $c=c(x)$ where the corresponding equation has at least nontrivial symmetries other than $\lie_1$. From \eqref{xieq} integration gives
\begin{equation}\label{xi}
  \xi(x,t)=\sqrt{a}\left(\frac{1}{2}\dot{\tau}I(x)+\rho(t)\right),
\end{equation}
where $\rho$ is arbitrary and $I(x)$ as in \eqref{IJ}. Substitution of $\xi$ into \eqref{IJ} and integration gives in terms of $I$ and $J$ (see \eqref{IJ})
\begin{equation}\label{phi}
  \phi(x,t)=-\frac{1}{8}\ddot{\tau}I^2-\frac{1}{2}\dot{\rho}I+\frac{1}{4}\dot{\tau} IJ+\frac{1}{2}\rho J+\sigma(t),
\end{equation}
where $\sigma(t)$ is a function of integration. Finally, the last determining equation of \eqref{deteqs} provides an equation for $\tau$, $\rho$, and $\sigma$
\begin{equation}\label{comp}
\frac{1}{8}\dddot{\tau}I^2-\frac{1}{4}\ddot{\tau}+\dot{\tau}\left(K(x)+\frac{1}{2}\sqrt{a}K'(x)I(x)\right)+\frac{1}{2}\ddot{\rho} I + \rho \sqrt{a} K'(x)  -\dot{\sigma}=0,
\end{equation}
or taking $I$ as the independent variable
\begin{equation}\label{comp-I}
  \frac{1}{8}\dddot{\tau}I^2-\frac{1}{4}\ddot{\tau}+\dot{\tau}\left(K+\frac{I}{2}\frac{dK}{dI}\right)+\frac{1}{2}\ddot{\rho} I + \rho \frac{dK}{dI}  -\dot{\sigma}=0,
\end{equation}
where $K$ is the semi-invariant defined in \eqref{inv}. It is remarkable that $K$ has reappeared in the computation of symmetries.  Using this invariant we are able to give a simple criterion for Eq. \eqref{mainLPE} to admit nontrivial symmetries  which are isomorphic to the Lie algebra of $G=\SL(2,\mathbb{R})\times \mathbb{R}$ or $G=\SL(2,\mathbb{R})\ltimes \Heis(3)$. Upon differentiating \eqref{comp-I} twice with respect to $I$ we obtain the classifying ODE
\begin{equation}\label{class-ODE}
   \frac{1}{4}\dddot{\tau}+\dot{\tau}\frac{d^2}{dI^2}\left(K+\frac{I}{2}\frac{dK}{dI}\right)+\rho\frac{d^3K}{dI^2}=0.
\end{equation}
We recall that for the Fokker-Planck equations we have
$$I(x)=\int\frac{dx}{\sqrt{p}},  \quad J(x)=-[(\sqrt{p})'+\frac{q}{\sqrt{p}}].$$ For this special class the symmetry analysis was presented in \cite{Hill1992}. Our $J$ differs from that of \cite{Hill1992} by a sign.

As we shall see below, the analysis of \eqref{class-ODE} will indicate that the necessary and sufficient conditions for Eq. \eqref{mainLPE} to possess $G=\SL(2,\mathbb{R})\times \mathbb{R}$ or $G=\SL(2,\mathbb{R})\ltimes \Heis(3)$ as symmetry groups is that the semi-invariant $K(I)$ be equal to
\begin{equation}\label{4-dim}
F_4(I)=\frac{\mu}{I^2}+c_2 I^2+c_0,\quad \mu\ne 0
\end{equation}
or
\begin{equation}\label{6-dim}
F_6(I)=c_2 I^2+c_1 I+c_0,
\end{equation}
respectively.

For the potential heat equation
\begin{equation}\label{poten}
 u_t=u_{xx}+V(x,t)u,
\end{equation}
we put $a=1$, $b=0$ and $c=V(x,t)$ in \eqref{deteqs} and find
$$\xi(x,t)=\frac{1}{2}\dot{\tau}x+\rho(t), \quad \phi(x,t)=-\frac{1}{8}\ddot{\tau}x^2-\frac{1}{2}\dot{\rho}x+\varsigma(t)-\frac{1}{4}\dot{\tau},$$
where the potential satisfies
\begin{equation}\label{potential-heat}
  \tau V_t+\xi V_x+\dot{\tau}V+\frac{1}{8}\dddot{\tau}x^2+\frac{1}{2}\ddot{\rho}x-\dot{\varsigma}=0.
\end{equation}
We note that we have defined $\varsigma=\sigma+\dot{\tau}/4$.
For given $V$, this equation provides a relation for the determination of the functions $\tau(t)$, $\rho(t)$ and $\varsigma(t)$. In the special case $V=V(x)$ we have $I=x$, $J=0$ and $K=V(x)$ and by  \eqref{6-dim} the symmetry algebra becomes six-dimensional when $V(x)$ is a quadratic function. For $V=K=0$ from either \eqref{comp-I} or \eqref{potential-heat} we reobtain the symmetry algebra of the heat equation. We know from the previous discussions that Eq. \eqref{poten} with potential $V(x,t)=q_2(t)x^2+q_1(t)x+q_0(t)$ is locally equivalent to the one with zero potential. In this case, \eqref{potential-heat} is split into the following set of linear ODEs for $\tau$, $\rho$ and $\varsigma$
\begin{equation}\label{tau-rho-sig}
  \begin{split}
      & \dddot{\tau}+16q_2 \dot{\tau}+8\dot{q}_2\tau=0, \\
       & \ddot{\rho}+4q_2\rho=-[3q_1 \dot{\tau}+2\dot{q}_1\tau],  \quad \dot{\varsigma}=\frac{d}{dt}(q_0\tau)+q_1\rho.
   \end{split}
\end{equation}
Since the general solutions of these equations  depend on 6 arbitrary parameters, the corresponding 6-dimensional symmetry algebra will be isomorphic to that of the  heat equation with $V=0$.

We shall present a detailed study of the infinitesimal symmetries for Eq. \eqref{mainLPE} in full generality. The cases $\rho=0$ and $\rho\ne 0$ will be considered separately.
\subsection{4-dimensional symmetry algebra}\label{Subsection4-dim}
$\rho=0$: From \eqref{class-ODE} we have
$$-\frac{\dddot{\tau}}{4\dot{\tau}}=\frac{d^2}{dI^2}\left(K+\frac{I}{2}\frac{dK}{dI}\right)=\text{const.}\equiv 4c_2.$$ Integration of  the second equation above gives
$$K=c_2 I^2+c_1I+c_0+\frac{\mu}{I^2}.$$ The first equation is $\dddot{\tau}+16c_2\dot{\tau}=0$.
Substitution of $K(I)$ into Eq. \eqref{comp-I} gives
\begin{equation}\label{splitcomp4}
  c_1=0,  \quad \dot{\sigma}=-\frac{1}{4}\ddot{\tau}+c_0\dot{\tau}.
\end{equation}
So the form of $K$ in \eqref{4-dim} is obtained.
The general solution of $\tau$ and $\sigma$ depends on four arbitrary constants and hence the symmetry algebra is 4-dimensional when $\mu\ne 0$.

Depending on the sign of $c_2$ we have three possible different solutions for $\tau$. We intend to present all symmetry vector fields in some basis.
\begin{enumerate}
\item $c_2=0$.
\begin{equation}\label{4-dim-c20}
\begin{split}
&v_1=T=\gen t,\\
&v_2=t\gen t+\frac{1}{2}\sqrt{a}I\gen x+(c_0t+\frac{1}{4}IJ)u\gen u,\\
&v_3=t^2\gen t+t\sqrt{a}I\gen x+\frac{1}{4}[2(2c_0t-1)t-I^2+2tIJ]u\gen u,\\
&v_4=M=u\gen u.
\end{split}
\end{equation}
The non-zero commutation relations are
\begin{equation}\label{comm-4-dim-c20}
\begin{split}
&[v_1, v_2]=v_1+c_0v_4,\\
&[v_1, v_3]=2v_2-\frac{1}{2}v_4, \quad [v_2, v_3]=v_3.
\end{split}
\end{equation}

\item $c_2=-\kappa^2$, $\kappa>0$.
\begin{equation}\label{4-dim-c2n}
\begin{split}
&v_1=T=\gen t,\\
&v_2=e^{4\kappa t}\gen t+2\kappa e^{4\kappa t}\sqrt{a}I\gen x-e^{4\kappa t}(-c_0+\kappa+2\kappa^2I^2-\kappa IJ)u\gen u,\\
&v_3=e^{-4\kappa t}\gen t-2\kappa e^{4\kappa t}\sqrt{a}I\gen x+e^{4\kappa t}(c_0+\kappa-2\kappa^2I^2-\kappa IJ)u\gen u,\\
&v_4=M=u\gen u.
\end{split}
\end{equation}
The non-zero commutation relations are
\begin{equation}\label{comm-4-dim-c2n}
\begin{split}
&[v_1, v_2]=4\kappa v_2,\\
&[v_1, v_3]=-4\kappa v_3, \quad [v_2, v_3]=-8\kappa v_1-8c_0 \kappa v_4.
\end{split}
\end{equation}

\item $c_2=\kappa^2$, $\kappa>0$.
\begin{equation}\label{4-dim-c2p}
\begin{split}
&v_1=T=\gen t,\\
&v_2=\cos 4\kappa t\gen t-2\kappa \sin 4\kappa t \sqrt{a}I\gen x+[c_0 \cos 4\kappa t+2 \kappa^2 \cos 4\kappa tI^2\\
& \quad +\kappa \sin 4\kappa t(1- IJ)]u\gen u,\\
&v_3=\sin 4\kappa t\gen t+2\kappa \cos 4\kappa t \sqrt{a}I\gen x+[c_0 \sin 4\kappa t+2 \kappa^2 \sin 4\kappa tI^2\\
& \quad -\kappa \cos 4\kappa t(1- IJ)]u\gen u,\\
&v_4=M=u\gen u.
\end{split}
\end{equation}
The non-zero commutation relations are
\begin{equation}\label{comm-4-dim-c2p}
\begin{split}
&[v_1, v_2]=-4\kappa v_2,\\
&[v_1, v_3]=4\kappa v_3, \quad [v_2, v_3]=4\kappa v_1+4c_0 \kappa v_4.
\end{split}
\end{equation}

\end{enumerate}

\begin{rmk}\label{rmk-lie4}
The algebra with basis \eqref{comm-4-dim-c20} is isomorphic to the direct sum $\lie_4=\Sl(2,\mathbb{R})\oplus \lr{M}$ which is easily seen by a change of basis
$$v_1\to v_1+c_0v_4,  \quad v_2\to 2v_2-\frac{1}{2}v_4, \quad v_3\to v_3.$$
The same isomorphism is also true for the algebras spanned by \eqref{4-dim-c2n} and \eqref{4-dim-c2p} which is achieved by
$$v_1\to (4\kappa)^{-2}(-8\kappa v_1-8c_0 \kappa v_4), \quad v_2\to (4\kappa)^{-1}v_2, \quad v_3\to (4\kappa)^{-1}v_3.$$
\end{rmk}

\subsection{6-dimensional symmetry algebra}\label{Subsection6-dim}
$\rho\ne 0$: From \eqref{class-ODE} we find that the following equations must be compatible
$$(1+\frac{I}{2})\dddot{K}+\ddot{K}=k_1,  \quad \dddot{K}=k_0,$$ where $k_0$, $k_1$ are some constants and the dot denotes derivative with respect to the argument $I$. This is possible if $k_0=0$ and $k_1$ arbitrary which we choose $k_1=-4c_2$. This means that $K$ should be quadratic in $I$ as in \eqref{6-dim}.

Splitting \eqref{comp-I} for this choice of $K$  provides the following equations (Compare with \eqref{tau-rho-sig})
\begin{equation}\label{splitcomp6}
\dddot{\tau}+16c_2\dot{\tau}=0,  \quad \ddot{\rho}+4c_2 \rho=-3 c_1\dot{\tau},  \quad \dot{\sigma}=-\frac{1}{4}\ddot{\tau}+c_0\dot{\tau}+c_1 \rho, \quad \rho\ne 0.
\end{equation}
The general solution of this system will depend on 6 arbitrary independent constants which lead to the following bases
for the corresponding algebras.

\begin{enumerate}
\item $c_2=0$.
\begin{equation}\label{6-dim-c20}
\begin{split}
&v_1=T=\gen t,\\
&v_2=t\gen t+\frac{1}{2}\sqrt{a}(I-3c_1t^2)\gen x+\frac{1}{4}[I(6c_1 t+J)+t(4c_0-2c_1^2t^2-3c_1tJ)]u\gen u,\\
&v_3=t^2\gen t+t\sqrt{a}(I-c_1 t^2)\gen x-\frac{1}{4}[I^2-2tI(3c_0 t+J)+t(2-4c_0 t+c_1^2 t^3+2c_1 t^2 J)]u\gen u,\\
&v_4=t\sqrt{a}\gen x+\frac{1}{2}[-I+t(c_1 t+J)]u\gen u,\\
&v_5=\sqrt{a}\gen x+\frac{1}{2}[2c_1t+J]u\gen u,\\
&v_6=M=u\gen u.
\end{split}
\end{equation}
The non-zero commutation relations are
\begin{equation}\label{comm-6-dim-c20}
\begin{split}
&[v_1, v_2]=v_1-3c_1 v_4+c_0v_6, \quad [v_1, v_3]=2v_2-\frac{v_6}{2}, \quad [v_1, v_4]=v_5,\\
&[v_1, v_5]=2c_1 v_6, \quad [v_2, v_3]=v_3, \quad [v_2, v_4]=\frac{1}{2}v_4, \quad [v_2, v_5]=-\frac{1}{2}v_5,\\
&[v_3, v_5]=-v_4,  \quad [v_4, v_5]=\frac{v_6}{2}.
\end{split}
\end{equation}

\item $c_2=-\kappa^2$, $\kappa>0$.
\begin{equation}\label{6-dim-c2n}
\begin{split}
&v_1=T=\gen t,\\
&v_2=e^{4\kappa t}\gen t+\frac{1}{\kappa} e^{4\kappa t}\sqrt{a}(-c_1+2\kappa^2 I)\gen x-\\
&\quad \frac{1}{4\kappa^2}e^{4\kappa t}[c_1^2-4c_0\kappa^2+4\kappa^3+8\kappa^4I^2+2c_1 \kappa J-4\kappa^2I(2c_1+\kappa J)]u\gen u,\\
&v_3=e^{-4\kappa t}\gen t+\frac{1}{\kappa} e^{-4\kappa t}\sqrt{a}(c_1-2\kappa^2 I)\gen x-\\
&\quad \frac{1}{4\kappa^2}e^{-4\kappa t}[c_1^2-4c_0\kappa^2-4\kappa^3+8\kappa^4I^2-2c_1 \kappa J+4\kappa^2I(-2c_1+\kappa J)]u\gen u,\\
&v_4=e^{2\kappa t}\sqrt{a}\gen x+\frac{1}{2\kappa}e^{2\kappa t}[(c_1-2\kappa^2 I)+\kappa J]u\gen u,\\
&v_5=e^{-2\kappa t}\sqrt{a}\gen x+\frac{1}{2\kappa}e^{-2\kappa t}[-(c_1-2\kappa^2 I)+\kappa J]u\gen u,\\
&v_6=M=u\gen u.
\end{split}
\end{equation}
The non-zero commutation relations are
\begin{equation}\label{comm-6-dim-c2n}
\begin{split}
&[v_1, v_2]=4\kappa v_2, \quad [v_1, v_3]=-4\kappa v_3, \quad [v_1, v_4]=2\kappa v_4,\\
&[v_1, v_5]=-2\kappa v_5, \quad [v_2, v_3]=-8\kappa v_1-\frac{2}{\kappa}p v_6, \quad [v_2, v_5]=-4\kappa v_4, \\
&[v_3, v_4]=4\kappa v_5,  \quad [v_4, v_5]=2 \kappa v_6,
\end{split}
\end{equation}
where $p=c_1^2+4c_0 \kappa^2$.

\item $c_2=\kappa^2$, $\kappa>0$.
\begin{equation}\label{6-dim-c2p}
\begin{split}
&v_1=T=\gen t,\\
&v_2=e^{4\kappa t}\gen t+\frac{1}{\kappa} e^{4\kappa t}\sqrt{a}(-c_1+2\kappa^2 I)\gen x-\\
&\quad \frac{1}{4\kappa^2}e^{4\kappa t}[c_1^2-4c_0\kappa^2+4\kappa^3+8\kappa^4I^2+2c_1 \kappa J-4\kappa^2I(2c_1+\kappa J)]u\gen u,\\
&v_3=e^{-4\kappa t}\gen t+\frac{1}{\kappa} e^{-4\kappa t}\sqrt{a}(c_1-2\kappa^2 I)\gen x-\\
&\quad \frac{1}{4\kappa^2}e^{-4\kappa t}[c_1^2-4c_0\kappa^2-4\kappa^3+8\kappa^4I^2-2c_1 \kappa J+4\kappa^2I(-2c_1+\kappa J)]u\gen u,\\
&v_4=e^{2\kappa t}\sqrt{a}\gen x+\frac{1}{2\kappa}e^{2\kappa t}[(c_1-2\kappa^2 I)+\kappa J]u\gen u,\\
&v_5=e^{-2\kappa t}\sqrt{a}\gen x+\frac{1}{2\kappa}e^{-2\kappa t}[-(c_1-2\kappa^2 I)+\kappa J]u\gen u,\\
&v_6=M=u\gen u.
\end{split}
\end{equation}
The non-zero commutation relations are
\begin{equation}\label{comm-6-dim-c2p}
\begin{split}
&[v_1, v_2]=-4\kappa v_3, \quad [v_1, v_3]=4\kappa v_2, \quad [v_1, v_4]=-2\kappa v_5,\\
&[v_1, v_5]=2\kappa v_4, \quad [v_2, v_3]=4\kappa v_1+r v_6, \quad [v_2, v_4]=2\kappa v_5, \quad [v_2, v_5]=2\kappa v_4, \\
&[v_3, v_4]=-2\kappa v_4, \quad [v_3, v_5]=2\kappa v_5,  \quad [v_4, v_5]=- \kappa v_6,
\end{split}
\end{equation}
where $r=4c_0 \kappa-\frac{c_1^2}{\kappa}$.

\end{enumerate}

\subsection{Lie-algebraic structure of the symmetry algebras}\label{structure}
The Lie algebras obtained in \ref{Subsection4-dim} and \ref{Subsection6-dim} appear in some nonstandard basis. We can transform them to known algebras. For example, the Lie algebra with the basis \eqref{comm-4-dim-c20} is isomorphic to the direct sum $\lie_4=\Sl(2,\mathbb{R})\oplus \lr{M}$ which is easily seen by a change of basis
$$v_1\to v_1+c_0v_4,  \quad v_2\to 2v_2-\frac{1}{2}v_4, \quad v_3\to v_3.$$
The same isomorphism is also true for the algebras spanned by \eqref{4-dim-c2n} and \eqref{4-dim-c2p} which is achieved by changing the basis
$$v_1\to (4\kappa)^{-2}(-8\kappa v_1-8c_0 \kappa v_4), \quad v_2\to (4\kappa)^{-1}v_2, \quad v_3\to (4\kappa)^{-1}v_3.$$

One can also see that the symmetry algebras in cases \eqref{6-dim-c20}, \eqref{6-dim-c2n} and  \eqref{6-dim-c2p} can be written as  a Levi-decomposition of the form $\lie_6=\Sl(2,\mathbb{R})\vartriangleright \nil(3)$, where $\nil(3)$ is the nilradical (Heisenberg algebra) with $M=v_6$ being the center element. This is clearly seen from the commutation relations in the first two cases. In the last case, the commutation relations between $v_1$, $v_2$, $v_3$ show that $\lr{v_1, v_2, v_3}$ is a pseudo-orthogonal Lie algebra $\Orr(2,1)$. To see this, we first transform $v_1\to 4\kappa v_1+r v_6$ followed by a scaling of the elements. The simple algebra $\Orr(2,1)$ is  isomorphic to $\Sl(2,\mathbb{R})$ which is realized by the change of basis
$$v_1\to v_1+v_2, \quad v_2\to v_3, \quad v_3\to v_1-v_2.$$
We conclude that the maximal finite-dimensional symmetry algebra $\lie_6$ is isomorphic to that of the heat equation $u_t=u_{xx}$ (see  \eqref{Heat-Symm}) and there should exist a point transformation transforming $\lie_6$ to \eqref{Heat-Symm}.
To construct such a transformation one starts with two commuting elements of the full algebra and transform to $\lr{\gen {\tilde{t}}, \gen {\tilde{x}}}$ by the equivalence group and then  the freedom left is used to appropriately transform the remaining basis elements.

We sum up our results as a theorem.
\begin{thm}
The dimension of the nontrivial symmetry algebra of Eq. \eqref{mainLPE}  is either four or six. A four-dimensional symmetry algebra occurs if and only if
\begin{equation}\label{L4}
  K(x)=\frac{1}{2}\sqrt{a}J'(x)-\frac{1}{4}J^2(x)+c(x)=\frac{\mu}{I^2}+c_2 I^2+c_0,\quad \mu\ne 0,
\end{equation}
a six-dimensional one (maximal) if and only if
\begin{equation}\label{L6}
  K(x)=c_2 I^2+c_1 I+c_0
\end{equation}
for some constants $c_2$, $c_1$, $c_0$, $\mu$ and the functions $I$, $J$ as defined in \eqref{IJ}. The first algebra  generated by \eqref{4-dim-c20},  \eqref{4-dim-c2n}, \eqref{4-dim-c2p} is isomorphic to $\lie_4=\Sl(2,\mathbb{R})\oplus \lr{M}$, the second one generated by \eqref{6-dim-c20},  \eqref{6-dim-c2n}, \eqref{6-dim-c2p} isomorphic to the Schr\"odinger (or heat) algebra.
\end{thm}

\subsection{Lie symmetries in $n$ dimensions}
For the sake of completeness we also present the two canonical forms of Eq. \eqref{Heatn} and their symmetry algebras.
We know from Subsection \eqref{Sectionheatn} that
$$u_t=\Delta u+\left[a(t)|x|^2+\sum_{k=1}^n b_k(t) x_k+\sum_{k=1}^n c_k(t)\right]u$$ is equivalent to the standard heat equation $u_t=\Delta u$ under the equivalence transformations. They have isomorphic symmetry groups having the structure $G=\Schr(n)=(\SL(2,\mathbb{R})\times \SO(n)) \ltimes  \Heis(n)$, where $\Heis(n)$ is the $(2n+1)$-dimensional Heisenberg group and $\dim G=n(n-1)/2+4+2n=(n^2+3n+8)/2$.  The Lie algebra of $G$  in the standard basis (see for example \cite{FushchichShtelenSerov1993}) is spanned by
\begin{equation}\label{Lie-n}
\begin{split}
& T=\gen t,  \quad D=2t \gen t+\sum_{k=1}^{n}x_k \gen{x_{k}}, \quad C=t^2\gen t+t\sum_{k=1}^{n}x_k \gen{x_{k}}-\frac{1}{4}(|x|^2+2nt)u\gen u,\\
& J_{kl}=x_k\gen{x_{l}}-x_l\gen{x_{k}},  \quad P_k=\gen{x_{k}}, \quad B_k=t\gen{x_{k}}-\frac{x_k u}{2}\gen u,  \quad M=u\gen u,  \quad k,l=1,2,\ldots,  n.
\end{split}
\end{equation}
These symmetries were originally obtained by Goff \cite{Goff1927} in 1927.

We turn again to the $n$-dimensional heat equation with potential $V(x,t)$ of \eqref{Heatn}.
\begin{thm}

The Lie point symmetries of Eq. \eqref{Heatn} are generated by  vector fields
\begin{subequations}\label{Heat-Symm-n}
\begin{equation}
  \mathbf{v}=\sum_{\mu=0}^{n}\xi_{\mu}(x)\gen \mu+\eta(x,u)\gen u=\tau(t)\gen t+\sum_{k=1}^{n}\xi_k(x,t)\gen{x_{k}}+\phi(x,t)u\gen u,
\end{equation}
where
\begin{equation}
  \xi_k=\frac{1}{2}\dot{\tau}a_{kl}x_l+\rho_k(t), \quad
\phi(x,t)=-\frac{1}{8}\ddot{\tau}|x|^2-\frac{1}{2}\dot{\rho}_k(t)a_{kl}x_l+
\sigma(t)-\frac{n}{4}\dot{\tau}.
\end{equation}
\end{subequations}
Here $a_{kl}=-a_{lk}$, i.e. $a_{kl}\in \So(n)$ and the functions $\tau(t)$, $\rho_k(t)$, $\sigma(t)$ and the constants $a_{kl}$ depend on the potential and satisfy
\begin{equation}\label{potential-eq-n}
  \tau V_t+(\mathbf{\xi}\cdot{\mathbf{\nabla}_x})V+\dot{\tau}V+\frac{1}{8}\dddot{\tau}|x|^2
+\frac{1}{2}\ddot{\rho}_k(t)a_{kl}x_l-\dot{\sigma}=0.
\end{equation}
\end{thm}
\begin{proof}
One way is to apply Lie's standard algorithm. However, it is easier by Proposition \eqref{pro3}. Indeed, we shall use the fact that a point symmetry of \eqref{Heatn} is an equivalence transformation with the property $\tilde{u}_{\tilde{t}}=\Delta \tilde{u}+V(\tilde{x},\tilde{t})\tilde{u}$ whenever $u_t=\Delta u+V(x,t)u$. So the symmetry vector field $\mathbf{v}$ is the infinitesimal generator of a local one-parameter group of equivalence transformations. If we allow $X_{\mu}$, $\mu=0,1,\ldots n$ and $U$ to depend on a parameter $\varepsilon$
$$
\tilde{x}_{\mu}(x; \varepsilon)=X_{\mu}(x; \varepsilon), \quad
\tilde{u}(x,u; \varepsilon)=U(x,u; \varepsilon),
$$
then the generators are easily found by differentiation  with respect to the parameter $\varepsilon$ at $\varepsilon=0$ as
$$
\xi_{\mu}(x)=\frac{d X_{\mu}}{d\varepsilon}(x; \varepsilon)\Bigr|_{\varepsilon=0},\qquad \eta(x,u)=\frac{d U}{d\varepsilon}(x, u; \varepsilon)\Bigr|_{\varepsilon=0}, \quad \mu=0,1,2,\ldots,n.
$$
Here we have
$\tau(t)=T'(t;\varepsilon)\bigr|_{\varepsilon=0}$, $\xi_{\sigma}=X'_{\sigma}(x;\varepsilon)\bigr|_{\varepsilon=0},$ where the prime denotes derivative with respect to $\varepsilon$ and from differentiation of  \eqref{mixed-der}, it follows
$$\xi_{\sigma, \rho}+\xi_{\rho, \sigma}=\dot{\tau}\delta_{\sigma \rho}, \quad \sigma, \rho=1,2,\ldots, n$$ and a similar expression for $\phi(x,t)=R'(x;\varepsilon)\bigr|_{\varepsilon=0}$ with $R(x;0)=1$ from \eqref{F}.
The symmetry condition \eqref{potential-eq-n} is obtained from differentiating equation \eqref{transformed-pot-n-2} with respect to $\varepsilon$ at $\varepsilon=0$.
\end{proof}
For a time-dependent rotationally invariant potential $V=V(|x|,t)=V(r,t)$,  rotations in addition to the trivial symmetries due to the linearity of the heat equation are always symmetries. Let us investigate all possible extensions of these symmetries. Eq. \eqref{potential-eq-n} can be written as follows:
\begin{equation}\label{potential-heat-n}
  \tau V_t+\frac{\dot{\tau}}{2}(rV_r+2 V)+\frac{\dddot{\tau}}{8}r^2+\sum_{k=1}^n\left(\frac{1}{2}\ddot{\rho}_k+\rho_k\frac{V_r}{r}\right)x_k-\dot{\sigma}=0.
\end{equation}
For a rotationally invariant potential $V(r)$, a consistent solution of  equation \eqref{potential-heat-n} is possible if $rV'+2 V=4Ar^2+2C$ for some constants $A$ and $C$. This implies $V(r)=A r^2+C+D/r^2$, $D$ being another constant. If $D=0$ we find
$$\dddot{\tau}+16A\dot{\tau}=0,  \quad \ddot{\rho_k}+4A\rho_k=0,  \quad \sigma=C \tau+\sigma_0.$$

For the more general time dependent (not rotationally invariant unless $B_k=0$)  potential
$$V(x,t)=A(t)|x|^2+\sum_{k=1}^{n}B_k(t)x_k+C(t),$$  splitting of the determining equation \eqref{potential-eq-n} gives
$$\dddot{\tau}+16A\dot{\tau}+8\dot{A}\tau=0, \quad \ddot{\rho_k}+4A\rho_k=-[3B_k\dot{\tau}+2\dot{B}_k\tau],  \quad \sigma=C \tau+\sigma_0. $$
$\tau$ satisfies is a self-adjoint third order linear equation with maximal symmetry and has the general solution $\tau=c_1\psi_1^2+c_2\psi_1\psi_2+c_3\psi_2^2$ in terms of the two independent solutions $\psi_1$ and $\psi_2$ of the linear oscillator equation $\ddot{\psi}+4A\psi=0$.
In this case  Lie symmetry algebra is isomorphic to $\Schr(n)$ and there is a point transformation taking $V$ to zero. The special choice $V=0$ of course leads to \eqref{Lie-n}.

On the other hand, if $D\ne 0$, then we should have $\rho_k=0$. Consequently,
$$u_t=\Delta u+\left(A |x|^2 +C+\frac{D}{|x|^2}\right)u, \quad D\ne 0$$ will be invariant under the group isomorphic to $G=\SL(2,\mathbb{R})\times \Or(n) \times \mathbb{R}$, where $\mathbb{R}$ is the abelian group of reals generated by time translations (trivial symmetries are excluded). It is worth remarking that in the absence of rotational symmetry, invariance under the symmetry $G=\SL(2,\mathbb{R})\times \mathbb{R}$ generated by the vector fields $\curl{T, D, C}$ of \eqref{Lie-n} results in a homogenous potential $V(x)$ of degree $-2$:
$$V(x)=V(x_1,\ldots,x_n)=x_1^{-2}H\left(\frac{x_2}{x_1},\ldots,\frac{x_n}{x_1}\right).$$

We refer to \cite{Thieullen2002} for the Lie point symmetries of the most general multidimensional parabolic equations
$$u_t+\sum_{i,j=1}^n a^{ij}(x,t)u_{x_ix_j}+\sum_{i=1}^n b^{i}(x,t)u_{x_i}+c(x,t)u=0$$
on a Riemannian manifold $\mathbb{R}^n$ equipped with the metric defined by the matrix $a^{-1}$ and also to \cite{0305-4470-32-48-310}.
For 2+1-dimensional FP equations  in the special case where the diffusion matrix is constant and the drift vector is irrotational, a complete symmetry classification was given in \cite{0305-4470-32-14-008}.

\subsection{Applications}
We have seen in Section \eqref{Section2} that any equation of the form
$$u_t=u_{xx}+(c_2 x^2+c_1 x+c_0)u$$
are locally equivalent to the heat equation under the transformation \eqref{ToHeatTr}. This means they should have symmetry algebra isomorphic to the heat algebra $\lie=\Sl(2,\mathbb{R})\vartriangleright \nil(3)$. As a reference we will list vector fields spanning $\lie$ for the heat equation with linear potential and the quadratic potential (harmonic oscillator equation).
\begin{exa}
$$u_t=u_{xx}-x u.$$

We have $I=x$, $J=0$, $K=c(x)=-x$ so that we put $c_2=c_0=0$, $c_1=-1$  in formula \eqref{6-dim-c20} and find the basis
\begin{equation}
\begin{split}
& v_1= \gen t,\\
& v_2=t \gen t+\frac{1}{2}(x+3t^2)\gen x-\frac{1}{2}(3xt+t^3)M,\\
& v_3=t^2 \gen t+(xt+t^3)\gen x-\left[\frac{x^2}{4}+\frac{3x^2t}{2}+\frac{t^4}{4}+\frac{t}{2}\right]M,\\
& v_4=\gen t +\gen x-tM,\\
& v_5=t\gen x-\frac{1}{2}(x+t^2)M,\\
& v_6=u\gen u=M.
\end{split}
\end{equation}
\end{exa}

\begin{exa}
\begin{equation}\label{Harmonic-Oscillator}
  u_t=u_{xx}+x^2 u.
\end{equation}

We have $I=x$, $J=0$, $K=c(x)=x^2$ so that we put $c_1=c_0=0$, $c_2=1$ in formula \eqref{6-dim-c20} and find the basis
\begin{equation}
\begin{split}
& v_1= \gen t,\\
& v_2= \cos 4t\gen t-2\sin 4t x\gen x+(2\cos 4t x^2+\sin 4t)M,\\
& v_3= \sin 4t\gen t+2\cos 4t x\gen x+(2\sin 4t x^2-\cos 4t)M,\\
& v_4=\sin 2t \gen x-\cos 2t x M,\\
& v_5=\cos 2t \gen x+\sin 2t x M,\\
& v_6=u\gen u=M.
\end{split}
\end{equation}
A variant of the harmonic oscillator equation \eqref{Harmonic-Oscillator}
$$u_t=u_{xx}-(x^2-1) u$$
has been used to construct Mehler's kernel by transforming the time independent solution $u_0(x)=e^{-x^2/2}$ by a symmetry group (generated by $v_3$ of \eqref{Harmonic-Oscillator-variant-basis}) followed by a  group (generated by $v_4-v_5$ of \eqref{Harmonic-Oscillator-variant-basis}) equivalent to the space translation
(the details are found in \cite{CraddockDooley2001}). Observe that there are no elementary solutions of the ODE (stationary oscillator equation) $u_{xx}\pm x^2u=0$. Its general solution is expressible in terms of parabolic cylinder functions. We simply put $c_2=-c_0=-1$, $c_1
=0$ in \eqref{6-dim-c2n} to find a basis for the symmetry algebra
\begin{equation}\label{Harmonic-Oscillator-variant-basis}
\begin{split}
& v_1= \gen t,\\
& v_2=  e^{4t}\gen t+2e^{4t} x\gen x-2e^{4t} x^2 M,\\
& v_3= e^{-4t}\gen t-2e^{-4t} x\gen x-2e^{4t} (x^2-1)M,\\
& v_4=e^{2t} \gen x-e^{2t} x M,\\
& v_5=e^{-2t} \gen x+e^{-2t} x M,\\
& v_6=u\gen u=M.
\end{split}
\end{equation}
\end{exa}

\begin{exa}[Ref. \cite{GazizovIbragimov1998}]\label{ex-BBS}
The back-propagating Black-Scholes equation with volatility $A$:
\begin{equation}\label{B-Black-Scholes}
  u_t+\frac{1}{2}A^2x^2u_{xx}+Bxu_x-Cu=0,  \quad A\ne 0.
\end{equation}
We have (after time reversal transformation $t\to -t$)
$$I=\frac{\sqrt{2}}{A} \ln x, \quad J=-\frac{\sqrt{2}}{A}\mathcal{D},  \quad K=-\frac{1}{2A^2}\mathcal{D}^2-C,\quad \mathcal{D}=B-\frac{A^2}{2}.$$
Since $K$ is a nonzero constant, our equation should have a 6-dimensional symmetry algebra.  We can use formula \eqref{6-dim-c20} with  $c_2=c_1=0$, $c_0=K$ to find a basis spanning the symmetry algebra. It can also be mapped to the heat equation by the general point transformation (we put $q_2=q_1=0$ and $q_0=-K$)
$$\tilde{t}=T(t), \quad \tilde{x}=\sqrt{\dot{T}}\left(\frac{\sqrt{2}}{A}\ln x+\omega_1 t+\omega_0\right),$$ $$u=\nu(t)x^{-\mathcal{D}/A^2+\dot{\delta}/(A\sqrt{2\dot{T}})}\exp[\frac{1}{2A^2}\frac{\ddot{T}}{\dot{T}}(\ln x)^2]\tilde{u},$$ where $\nu(t)$ is obtained from \eqref{nu} and $\delta(t)=\sqrt{\dot{T}}\omega(t)=\sqrt{\dot{T}}(\omega_1 t+\omega_0)$, $T(t)=(kt+m)/(lt+n)$, $kn-lm\ne 0$. Here $\omega_1, \omega_0, k, l, m, n$ are all arbitrary constants. The Black-Scholes transformation corresponds to the simplest choice $T=t$ (up to translation, the second transformation in \cite{GazizovIbragimov1998}):
\begin{equation}\label{BS-transform}
 \tilde{t}=t, \quad  \tilde{x}=\frac{\sqrt{2}}{A}\ln x+\omega_1 t+\omega_0, \quad u=\nu_0\exp[(K+\frac{\omega_1^2}{4})t]x^{-\frac{\mathcal{D}}{A^2}+\frac{\omega_1}{\sqrt{2} A}}\tilde{u}.
\end{equation}
We recover the first transformation of \cite{GazizovIbragimov1998} (which apparently contains a misprint) by choosing an inversional transformation $T=-1/t$ (up to a translation in $t$).

\end{exa}

\begin{exa}[Ref. \cite{CraddockLennox2007, CraddockLennox2009}]
\begin{equation}\label{drift-PDE}
  u_t=\sigma x u_{xx}+f(x)u_x-\mu x^{r}u.
\end{equation}
We have $a=\sigma x$, $b=f(x)$, $c=-\mu x^{r}$ and
$$I=\frac{2}{\sqrt{\sigma}}\sqrt{x},  \quad J=\frac{1}{\sqrt{\sigma x}}[-f+\frac{\sigma}{2}],  \quad K=-\frac{1}{2\sigma x}[\sigma x f'-\sigma f+\frac{1}{2}f^2+2\sigma \mu x^{r+1}].$$
We require to obtain at least $\lie_4$ or $\lie_6$ symmetry. So we form the equation
$$ K=c_2 I^2+c_1 I+c_0+\frac{m}{I^2},$$ which is explicitly written as
$$\sigma x f'-\sigma f+\frac{1}{2}f^2+2\sigma \mu x^{r+1}+\frac{3}{8}\sigma^2=-8 c_2 x^2-4 \sqrt{\sigma } c_1 x^{3/2} -2 \sigma c_0 x  -\frac{m \sigma ^2}{2}.$$
With the definition of the constants
$$A=-4\sqrt{\sigma}c_1, \quad B=-8 c_2, \quad C=-2\sigma c_0, \quad D=-\frac{\sigma^2}{8}(4m+3)$$
it has the form
\begin{equation}\label{Riccati4}
\sigma x f'-\sigma f+\frac{1}{2}f^2+2\sigma \mu x^{r+1}=Ax^{3/2}+B x^2+Cx+D.
\end{equation}
Solutions  of this Riccati equation will produce  a nontrivial symmetry group of dimension $d\in\curl{4,6}$ depending on whether $m\ne 0$ or $m=0$.
We have  $\lie_6$ symmetry algebra if $m=0$ ($3\sigma^2+8D=0$), otherwise $\lie_4$ one when $c_1=0$.
The particular choices  $c_1=0$ and $c_1=c_2=0$ lead to the Riccati equations
\begin{subequations}
\begin{equation}\label{drift-Riccati-1}
  \sigma x f'-\sigma f+\frac{1}{2}f^2+2\sigma \mu x^{r+1}=Ax^{2}+B x+C,
\end{equation}
and
\begin{equation}\label{drift-Riccati-2}
  \sigma x f'-\sigma f+\frac{1}{2}f^2+2\sigma \mu x^{r+1}=Ax+B.
\end{equation}
\end{subequations}
In the first case, $d=\dim \lie=4, 6$ depending on whether $m=-\frac{1}{4\sigma^2}(3\sigma^2+8C)\ne 0$, or  $3\sigma^2+8C=0$, respectively. In the second case, the splitting of the dimension is dictated by the condition $3\sigma^2+8B\ne 0$ or $3\sigma^2+8B=0$.
The vector fields of the $\lie_4$ and $\lie_6$ algebras can be directly obtained from the formulas of Subsections \eqref{Subsection4-dim} and \eqref{Subsection6-dim} depending on whether $c_2=-B/8$ or $c_2=0$ in the respective cases.

Note that the Cox-Ingersoll-Ross (CIR) PDE is included in the class \eqref{drift-PDE} where
$f(x)=a+bx$, $\mu=0$. The drift $f(x)$ satisfies \eqref{drift-Riccati-1} with
$$A=\frac{b^2}{2}, \quad B=ab, \quad C=\frac{a^2-2\sigma a}{2}, \quad b\ne 0.$$ Its symmetry algebra turns out  to be six-dimensional if $3\sigma^2-8a \sigma+4a^2=0$, otherwise four-dimensional. The first possibility implies that an invertible point transformation mapping the CIR model to the heat equation exists.  The  formula \eqref{ToHeatTr2} is used to construct such a transformation.

The more general form of \eqref{drift-PDE}, where
\begin{equation*}\label{fg}
  a(x)=\sigma x^{\gamma},  \quad b(x)=f(x), \quad c(x)=-g(x),  \quad x>0
\end{equation*}
is also very easily handled by forming the nontrivial symmetry  conditions \eqref{L4}-\eqref{L6} and using
$$I(x)=\frac{2}{\sqrt{\sigma}(2-\gamma)}x^{(2-\gamma)/2}, \quad J(x)=\frac{x^{(\gamma-2)/2} (\gamma  \sigma -2 h(x))}{2 \sqrt{\sigma }}, \quad h(x)=x^{1-\gamma}f(x).$$ A six-dimensional symmetry algebra occurs when
$$\sigma x h'-\sigma h+\frac{1}{2}h^2+2\sigma x^{2-\gamma}g=Ax^{\frac{3}{2}(2-\gamma)}+Bx^{4-2\gamma}+Cx^{2-\gamma}+D,$$
where
$$A=\frac{4\sqrt{\sigma}c_1}{\gamma-2},  \quad B=-\frac{8c_2}{(\gamma-2)^2},  \quad  C=-2\sigma c_0, \quad D=\frac{\gamma(\gamma-4)}{8}\sigma^2, \quad \gamma\ne 2.$$
For the special choice $\gamma=1$,  $g(x)=\mu x^r$, it is reduced to \eqref{Riccati4} with $m=0$.

The case $\gamma=2$ should be analysed separately. Using
$$I(x)=\frac{\ln x}{\sqrt{\sigma}},  \quad J(x)=\frac{x \sigma -f(x)}{\sqrt{\sigma }x}$$ we find the Riccati equation
$$\sigma x^2 f'+\frac{1}{2}f^2-2\sigma xf+2\sigma x^2g=Ax^2+Bx^2 \ln x+C (x\ln x)^2,$$ where
$$A=-2\sigma(c_0+\sigma),  \quad B=-2\sqrt{\sigma}c_1,  \quad C=-2c_2$$ as the condition for the existence of heat algebra ($\lie_6$ symmetry algebra).

The corresponding bases are immediately at our disposal  from the results of Subsection \eqref{Subsection6-dim}.

Following the same procedure, the Lie point symmetries of all equations studied in Refs. \cite{CraddockDooley2001, CraddockPlaten2004, CraddockLennox2009,  CraddockKonstandatosLennox2009, BaldeauxPlaten2013} can also be recovered with minimal computational effort.
\end{exa}

A special case where $f$ and $g$ of \eqref{fg} are some power functions was studied in \cite{SinkalaLeachOHar2008}.

\begin{exa}[\cite{SinkalaLeachOHar2008}]
$$u_t=-\frac{1}{2}\rho^2 x^{2\gamma}u_{xx}+[\alpha+\beta x-\lambda \rho x^{\gamma}]u_x-xu,  \quad \gamma\ne 1, \; \rho\ne 0.$$
Here
$$I(x)=\frac{\sqrt{2} x^{1-\gamma }}{(1-\gamma ) \rho }, \quad J(x)=\frac{\sqrt{2} x^{-\gamma }}{\rho } \left(-\alpha -x \beta +x^{\gamma } \lambda  \rho +\frac{1}{2} x^{2 \gamma-1 } \gamma  \rho ^2\right).$$ Now we can set up our condition for the existence of nontrivial symmetry
\begin{equation}\label{symcond}
\begin{split}
    K(x)&=x^{-1}\left[-x^2+\alpha  \gamma + \left(-\frac{\beta }{2}+\beta  \gamma -\frac{\lambda ^2}{2}\right)x+\frac{\alpha  \lambda x^{1-\gamma } }{\rho }-\frac{\alpha ^2x^{1-2 \gamma } }{2 \rho ^2}-\frac{\alpha  \beta x^{2-2 \gamma }  }{\rho ^2}-\right.\\
   &\left.-\frac{\beta ^2 x^{3-2 \gamma } }{2 \rho ^2}+\frac{\beta  \lambda  x^{2-\gamma } }{\rho }-\frac{1}{2}\gamma  \lambda  \rho x^{\gamma }  + \left(-\frac{\gamma  \rho ^2}{4}+\frac{\gamma ^2 \rho ^2}{8}\right)x^{-1+2 \gamma }\right] \\
    & =\frac{2 c_2 x^{2-2 \gamma }}{(1-\gamma )^2 \rho ^2}+\frac{\sqrt{2} c_1 x^{1-\gamma }}{(1-\gamma ) \rho }+\frac{\mu \rho^2}{2}(\gamma-1)^2x^{-2+2 \gamma }+c_0.
\end{split}
\end{equation}
Matching  different powers of $x$ reveals that $\gamma$ can take values in the set $\curl{0, 1/2, 3/2, 2}$. A study of symmetries for all possible $\gamma$ was given in \cite{SinkalaLeachOHar2008}. We shall only consider the case $\gamma=2$ to illustrate how the results of \cite{SinkalaLeachOHar2008} can be recovered with our approach.   We put $\gamma=2$ in \eqref{symcond} and split with respect to  the powers of $x$  we find
$$\alpha=0, \quad \mu=0, \quad 1+\lambda \rho=0, \quad 4 c_2+\beta ^2=0, \quad \beta =\sqrt{2} c_1 \rho , \quad (-8 c_0+12 \beta ) \rho ^2=4.$$
So the last three equations can be solved for $c_2$, $c_1$ and $c_0$ in terms of $\beta$, $\rho$ with the condition $\lambda \rho=-1$ as
$$c_2=-\frac{\beta^2}{4},  \quad c_1=\frac{\beta}{\sqrt{2}\rho},  \quad c_0=\frac{1}{2\rho^2}(3\beta \rho^2-1).$$
In summary, we have shown that for the special choice $\gamma=2$, $\alpha=0$, $\lambda \rho=-1$ the symmetry algebra is isomorphic to $\lie^{(6)}$ of \eqref{Heat-Symm}. The formula \eqref{6-dim-c2n} is at our disposal for the vector fields of the symmetry algebra. We shall compute only two of them for illustration. Using the following quantities
$$I=-\frac{\sqrt{2}}{\rho}x^{-1}, \quad J=\frac{\sqrt{2}}{\rho}x^{-1}\left(\rho ^2 x^2 -x-\beta \right), \quad IJ=\frac{2 \beta }{\rho ^2 x^2 }+\frac{2}{\rho ^2 x }-2,\quad \kappa=\frac{\beta}{2},$$ we immediately find (after performing time reflection $t\to -t$)
\begin{eqnarray*}
  v_4 &=& \frac{\rho}{\sqrt{2}}e^{-\beta t}[x^2 \gen x+x u\gen u], \\
  v_5 &=& \frac{\rho}{\sqrt{2}}e^{\beta t}\left[x^2 \gen x-\left(\frac{2 \beta }{\rho ^2 x }+\frac{2}{\rho ^2  }-x\right)u\gen u\right],
\end{eqnarray*}
which, up to a factor of $\rho/\sqrt{2}$, are exactly $G_3$, $G_4$ of \cite{SinkalaLeachOHar2008}. Of course, there is no need to transform the generators to construct the point transformation to the heat equation as was done in \cite{SinkalaLeachOHar2008}. It suffices to use the general transformation formula \eqref{ToHeatTr2} with $(q_2, q_1, q_0)=(c_2, c_1, c_0)$.
Finally, we take $\gamma=3/2$. Eq. \eqref{symcond} implies that
$$\alpha=0, \quad \lambda=0, \quad c_1=0, \quad \mu\ne 0$$ and $c_2$, $c_0$, $\mu$ can be solved
$$c_2=-\frac{\beta^2}{16}, \quad c_0=\beta,  \quad \mu=-\frac{1}{4\rho^2}(3\rho^2+32), \quad \kappa=\frac{\beta}{4}.$$
The fact that $c_1=0$, $\mu\ne 0$ implies that we have $\lie_4$ symmetry. Two additional symmetries $v_2$, $v_3$ are easily computed from \eqref{4-dim-c2n}  using the quantities
$$\sqrt{a}I=-2x,  \quad I=-\frac{2\sqrt{2}}{\rho}x^{-1/2}, \quad IJ=\frac{4 \beta }{\rho ^2 x }-3,$$ which lead to
\begin{eqnarray*}
  v_2 &=& -e^{-\beta t}[\gen t +\beta x \gen x], \\
  v_3 &=& e^{\beta t}[-\gen t+\beta x \gen x+2\beta\left(1-\frac{\beta }{x \rho ^2}\right)u\gen u],
\end{eqnarray*}
which coincides with $G_4, G_3$ of \cite{SinkalaLeachOHar2008}, up to a constant multiple.
Similarly, based on the expression \eqref{symcond} all other cases of $\gamma$ can be treated without any difficulty. It is clear that different values of $\gamma$ will impose different restrictions on the parameters $\alpha, \beta, \lambda, \rho$ of the equation.
\end{exa}

\begin{exa}[Ref. \cite{Lehnigk1989, Pesz2002}]
$$u_t+[(4x^2 \ln x)u_x+(Ax+B \ln x)u]_x=0,  \quad x\in[0,1].$$
Up to time reversal $t\to -t$, this is a Fokker-Planck equation with coefficients
$$p=(4x^2 \ln x), \quad q=Ax+B \ln x.$$ We have
$$\sqrt{p}I=2x{\ln x},  \quad J=-\left[\frac{d}{dx}(\sqrt{p})+\frac{q}{\sqrt{p}}\right],  \quad K=\frac{1}{2}J'-\frac{1}{4}J^2+q'.$$ Computation of $K$ gives
$$K=-\frac{1}{16}(A^2-4)\frac{1}{\ln x}-\frac{1}{16}(B-4)^2\ln x-\frac{1}{8}(A-4)(B-4).$$
For $A=\pm 2$, the equation has a $\lie_6$ symmetry and is equivalent to the heat equation. The transformation formula is easy to obtain from \eqref{ToHeatTr2}. Otherwise, the symmetry algebra would have to be of $\lie_4$ type. For instance, for the special case $A=B=0$, using the relations $IJ=-(1+2\ln x)$ and putting $c_2=-1$, $c_0=-2$ in \eqref{4-dim-c2n} we find
\begin{equation}
\begin{split}
&v_1=\gen t,\\
&v_2=e^{4t}(\gen t+4x \ln x \gen x),\\
&v_3=e^{-4t}[\gen t - 4x \ln x \gen x + 4(1+\ln x)u\gen u],\\
&v_4=u\gen u,
\end{split}
\end{equation}
as being the Lie symmetry algebra of $u_t+(4x^2\ln x u_x)_x=0$.

Using our criterion it can also be seen that
$$u_t=[b_1 x^2 u_x+(b_2 x+b_3 x \ln x)u]_x$$ has a $\lie_6$ symmetry algebra and hence there is a point transformation (formula \eqref{ToHeatTr2}) reducing this equation to the heat equation.
\end{exa}

\begin{rmk}
All equations of the type \eqref{mainLPE} arising in the literature can be analysed when the equation can have nontrivial symmetry algebra $\lie_4$ or $\lie_6$  by checking a simple invariant condition only. In the presence of $\lie_6$ symmetry the transformation formula to the heat equation is available in the general form. Furthermore, the corresponding  generators are given in some basis. See  Subsection \ref{structure}  for the structure of the algebras.
\end{rmk}

\subsection{Schr\"odinger equation}
All that has been said so far   about equivalence and symmetries of the parabolic equations Eq. \eqref{LPE} can be repeated for the complex linear parabolic equations of the form
\begin{equation}\label{complex-LPE}
  iu_t=a(x,t)u_{xx}+b(x,t)u_x+c(x,t)u, \quad x\in \Omega\subseteq\mathbb{R},\quad t>0,
\end{equation}
where $a$, $b$, $c$ and $u\in \mathbb{C}$. Indeed, any equation in the class with four- or six-dimensional symmetry algebra can be transformed to
\begin{equation}\label{2nd-cano-complex}
  iu_t=u_{xx}+\frac{\mu}{x^2}u,\quad \mu\ne 0,
\end{equation}
or to the free (with zero potential) Schr\"odinger equation
\begin{equation}\label{free-Schr-eq}
  iu_t=u_{xx}
\end{equation}
by an equivalence transformation of the form
\begin{equation}\label{EqTr-complex}
  \tilde{t}=T(t), \quad \tilde{x}=X(x,t),\quad u=Q(x,t)\tilde{u}(\tilde{x},\tilde{t}),
\end{equation}
where $Q$ is a complex-valued function.  The symmetry algebra of \eqref{free-Schr-eq} is known as the Schr\"odinger algebra. The linear heat and Schr\"odinger equations have isomorphic symmetry algebras. That is why, the symmetry algebra of the heat equation is sometimes called the Schr\"odinger algebra.
For a Schr\"odinger equation with a complex potential $V(x,t)=V_1(x,t)+iV_2(x,t)$ equivalence transformation is given by
\begin{equation}\label{EqTr-pot-Schr-eq}
  \tilde{t}=T(t), \quad \tilde{x}=\sqrt{\dot{T}}x+\rho(t), \quad u=R(t)\exp[i\omega(x,t)]\tilde{u}(\tilde{x},\tilde{t}),
\end{equation}
where
\begin{equation*}
\begin{split}
  &\omega(x,t)=-\frac{1}{2\sqrt{\dot{T}}}\left[\frac{d}{dt}\sqrt{\dot{T}}\;\frac{x^2}{2}+\dot{\rho}x+\chi(t)\right],\\
  & \tilde{V}=\frac{1}{\dot{T}}\left[V_1-\omega_t-\omega_x^2+i\left(V_2+\frac{\dot{R}}{R}+\omega_{xx}\right)\right]
\end{split}
\end{equation*}
with $\chi$ being an arbitrary function.

Results on the real line will also remain true for the multidimensional generalizations $iu_t=\Delta u$ of \eqref{complex-LPE}.
For completeness, we give the standard basis of the $n$-dimensional Schr\"odinger algebra
\begin{equation}\label{Schr}
  \schr(n)=\lr{T, \quad D, \quad C, \quad J_{kl}, \quad P_k \quad B_k, \quad M}, \quad k\ne l, \quad k,l=1,2,\ldots n,
\end{equation}
where
\begin{equation*}
  \begin{split}
      & T=\gen t,  \quad  D=2t\gen t+\sum_{a=1}^{n}x_{k} \gen{x_{k}}-\frac{n}{2}E,  \quad C=t^2\gen t+t \sum_{k=1}^{n}x_{k} \gen{x_{k}}-\frac{1}{4}(|x|^2M+2nt E), \\
       & J_{ab}=x_{k} \gen{x_{l}}-x_{l} \gen{x_{k}}, \quad P_k=\gen{x_{k}}, \quad B_k=t\gen{x_{k}}-\frac{x_k}{2}M, \quad M=i(u\gen u-u^{*}\partial_{u^{*}} ),        \end{split}
\end{equation*}
where $E=u\gen u+u^{*}\partial_{u^{*}}$. The algebra $\schr(3)$ was first obtained by Niederer \cite{Niederer1972}.

In case of Schr\"odinger equation with potential $V$, the Lie point symmetries are analogous to those \eqref{Heat-Symm-n}-\eqref{potential-eq-n} of the heat equation  and are generated by (see \cite{Niederer1974} and \cite{LeviTempestaWinternitz2004} for a sketchy proof)
$$\mathbf{v}=\tau(t)\gen t+\sum_{k=1}^{n}\xi_k(x,t)\gen{x_{k}}+i\phi(x,t)(u\gen u-u^{*}\gen {u^{*}}),$$
where
$$\xi_k=\frac{1}{2}\dot{\tau}x_k+\sum_{l=1}^{n}a_{kl}x_l+\rho_k(t), \quad
\phi(x,t)=-\frac{1}{8}\ddot{\tau}|x|^2-\frac{1}{2}\sum_{k=1}^{n}\dot{\rho}_k(t)x_k+
\sigma(t)+i\left[\frac{n\dot{\tau}}{4}+b\right].$$ Here $a_{kl}=-a_{lk}$ and $b$ are real constants. The real functions $\tau(t)$, $\rho_k(t)$, $\sigma(t)$ and the real constants $a_{kl}$ depend on the potential and satisfy
\begin{equation}\label{potential-Schr-n}
  \tau V_t+(\mathbf{\xi}\cdot{\mathbf{\nabla}_x})V+\dot{\tau}V+\frac{1}{8}\dddot{\tau}|x|^2
+\frac{1}{2}\sum_{k=1}^{n}\ddot{\rho}_k(t)x_k-\dot{\sigma}=0.
\end{equation}
The time-dependent Kepler potential of the form $V(r,t)=\tau(t)^{-1/2}r^{-1}$, $\tau(t)=\tau_2 t^2+\tau_1 t+\tau_0$ leads to the symmetry group that is the direct product of $\Or(n)$ with  one-dimensional subgroups of $\SL(2,\mathbb{R})$. For the static Kepler potential $V=\tau_0/r$, the point symmetry group is $G=\Or(n) \times \mathbb{R}$ (gauge symmetry is ignored).

Purely symmetric potentials
of the form $V(r)=A r^2+B+C/r^2$ allow nontrivial Lie point symmetries of dimension $(n^2-n+8)/2$ or $(n^2+3n+8)/2$ depending on whether $C\ne 0$ or $C=0$. For $A=B=0$, $C\ne 0$ we have the inverse square potential.  The fact that the symmetry group of the $n$-dimensional harmonic oscillator equation ($V(r)=A r^2$) is isomorphic to the  corresponding group of the free Schr\"odinger equation ($V=0$) was first established in \cite{Niederer1973}. A complete classification of time-independent potentials in dimensions $n=1,2,3$ was performed in \cite{Boyer1974}.

An attempt towards equivalence transformations for the  Schr\"odinger equation in $(n + 1)$-dimensions can be found in \cite{Schulze-Halberg2006}.

\section{Group-invariant and fundamental solutions}\label{Section4}
In order to be able to give a neat classification of solutions invariant under $\lie_6$ or $\lie_4$ of the Subsections \eqref{Subsection6-dim} and \eqref{Subsection4-dim} we need a classification of one-dimensional subalgebras under the adjoint transformations of the Lie  group of the symmetry algebras. We already know subalgebras of the heat equation symmetry algebra (up to isomorphism) $\Sl(2,\mathbb{R})\vartriangleright \nil(3)$ and  $\Sl(2,\mathbb{R})\oplus \mathbb{R}$ of the second canonical equation.

For the sake of completeness, we present a list of representatives of the subalgebras of the heat algebra given in \cite{Winternitz1989}:
\begin{equation}\label{subalgebras-lie-6}
  \begin{split}
     & \lr{D}, \quad  \lr{T}, \quad \lr{C+T}, \quad \lr{D+aM},    \\
     & \lr{T+B}, \quad  \lr{T+M}, \quad \lr{T- M},   \\
     & \lr{C+T+a M}, \quad \lr{P}, \quad \lr{M}, \quad a\in \mathbb{R}, \quad a\ne 0.
   \end{split}
\end{equation}
A similar classification can also be found in \cite{Olver1991}. We recall that the first attempt for a  subalgebra classification is due to Weisner \cite{Weisner1959}. Symmetry reductions and invariant solutions of the heat equation can be found in \cite{Olver1991}.
More specifically, the heat kernels for $u_t=\Delta u$ can be derived easily by either transforming the constant solution $u_0=1$ by nontrivial symmetries like dilational, Galilei, projective transformations  or by seeking the associated group-invariant solutions.

A list of  representatives of subalgebras of $\Sl(2,\mathbb{R})\oplus \mathbb{R}$ is given by
\begin{equation}\label{subalgebras-lie-4}
  \lr{T+a M}, \quad  \lr{D+a M}, \quad \lr{C+T+a M}, \quad \lr{M}, \quad a\in \mathbb{R}.
\end{equation}
We know that any equation of the form \eqref{mainLPE}  with a four-dimensional symmetry group can be transformed to its canonical form \eqref{canonical2}. It would be useful to discuss solutions invariant under the representative subalgebras \eqref{subalgebras-lie-4}. We recall that the group-invariant solutions of the heat equation satisfy parabolic cylinder equation (which can be transformed to Hermite or confluent hypergeometric equation) and Airy equation (which can be solved in terms of Bessel functions of index $1/3$), all of which belong to the generalized hypergeometric class.

\subsection{Group-invariant solutions of the second canonical equation}
We present symmetry reductions and invariant solutions of
\begin{equation}\label{2ndcanonical}
  u_t-u_{xx}=\frac{\mu}{x^2}u,  \quad \mu\ne 0
\end{equation}
based on the classification \eqref{subalgebras-lie-4}.
\begin{enumerate}
  \item Subalgebra $\lr{T+a M}$: The invariants are $x$ and $e^{-at}u$ so that we put $u=e^{at}F(x)$ to find invariant solutions. The reduced equation is
      \begin{equation}\label{reduced-1}
        x^2F''+(\mu-a x^2)F=0.
      \end{equation}
      The change of dependent variable $F=x H(x)$ transforms it to
      $$x^2H''+2xH'+(\mu-a x^2)H=0.$$ Setting $a=-\lambda^2$, its solution can be expressed in terms of Bessel functions $$H=x^{-1/2}Z_{\frac{\sqrt{1-4\mu}}{2}}(\lambda x),$$ where $Z$ stands for  Bessel functions depending on the sign of $\lambda$ and on whether the index $\sqrt{1-4\mu}/2$ being an integer or not. For the special choice $\mu=-n(n+1)$, $n\in \mathbb{Z}$ they are elementary. In this case, Eq. \eqref{reduced-1} is known to be Riccati-Bessel equation. Note that for $a=0$, it becomes Euler equation. Finally, we obtain the invariant solution
      $$u(x,t)=e^{-\lambda^2 t} x^{1/2}Z_{\frac{\sqrt{1-4\mu}}{2}}(\lambda x).$$
      The special case $\mu=-2$, $\lambda=1$ leads to the elementary solution
      $$u(x,t)=e^{-t}\sqrt{x}[c_1 J_{3/2}(\lambda x)+c_2 J_{-3/2}(\lambda x)],$$ which can be expressed as
      $$u=e^{-t}\left[c_1\left(\cos x-\frac{\sin x}{x}\right)+c_2 \left(\sin x+\frac{\cos x}{x}\right)\right].$$
  \item Subalgebra $\lr{D+a M}$: Using the invariants $z=xt^{-1/2}$, $t^{-a}u$ we set
  $u=t^a F(z)$ for the invariant solution which leads to the reduced equation
  $$F''+\frac{z}{2}F'+\left(\frac{\mu}{z^2}-a\right)F=0.$$ The change of independent variable $y=z^2/2$ reduces it   to
  $$2yF''(y)+(y+1)F'(y)+\left(\frac{\mu}{2y}-a\right)F(y)=0,$$
  which is a generalized hypergeometric equation and can be transformed to the confluent hypergeometric one
  $$2yH''+(2+\sqrt{1-4\mu}-y)H'-\frac{1}{4}(3+4a+\sqrt{1-4\mu})H=0$$
  by a linear change of dependent variable
  $F(y)=e^{-y/2}y^{1/4(1+\sqrt{1-4\mu})}H(y)$. The general solution of this equation  in terms of Kummer's confluent hypergeometric function $M$ is
  $$H(y)=c_1\; M(\alpha,\gamma,\frac{y}{2})+c_2\; y^{1-\gamma}M(\alpha-\gamma+1,2-\gamma,\frac{y}{2}),$$ where $\alpha=(3+4a+\sqrt{1-4\mu})/4$, $\gamma=1+(\sqrt{1-4\mu})/2$ and $\alpha,\gamma\ne 0, \pm1, \pm2,\ldots$ If $\gamma\in \mathbb{Z}$ both solutions coincide. Note that when $\alpha=-n$, $n\in \mathbb{N}$  the first solution becomes a generalized Laguerre polynomial $L_n^{\sqrt{1-4\mu}/2}(y)$ ($L_n^{\alpha}$ being defined by $L_n^{\alpha}(x)=\binom{n+\alpha}{n}M(-n,\alpha+1,x$). So the corresponding scale-invariant solutions become elementary. The simplest of them are obtained for $\alpha=0$ via the transformation
  $$F=e^{-y/2}y^{-(a+1/2)}c_0,$$ where $c_0$ is a constant and $\mu=-2(a+1)(2a+1)$. So we have obtained an elementary solution of the form
  \begin{equation}\label{elem-sol-1}
    u=c_0 x^{-(2a+1)}t^{2a+1/2}\exp[-\frac{x^2}{4t}].
  \end{equation}

  For $\mu=-2$, $a=-3/2$ it becomes $$u=c_0 x^2 t^{-5/2}\exp[-\frac{x^2}{4t}].$$

  Reduction to a confluent equation is not unique. We can also apply the change of variable $F=e^{-y/2}y^{1/4(1-\sqrt{1-4\mu})}H(y)$. Now we have $\alpha=(3+4a-\sqrt{1-4\mu})/4$. Another elementary solution is
  $$u=c_0 x^{-(2a+1)}t^{1/2}\exp[-\frac{x^2}{4t}]$$ with $\mu=-2(a+1)(2a+1)$ ($\alpha=0$).

  \item Subalgebra $\lr{C+T+a M}$: We look for an invariant solution of the form
  $$u=(1+t^2)^{-1/4}\exp\left[-\frac{tx^2}{4(1+t^2)}+\mu \arctan t\right]F(z), \quad z=x^2(1+t^2)^{-1}.$$ The reduced equation is a generalized hypergeometric equation
  $$F''+\frac{2}{4z}F'+\frac{z^2-4az+4\mu}{(4z)^2}F=0.$$  A further transformation can be applied to transform it to a confluent hypergeometric equation. The transformation  $$F=e^{-i z/4}\eta^{1/4(1-i\sqrt{4\mu-1})}H(z)$$
  reduces it to
  $$8z H''+4i(\sqrt{4\mu-1}-2i-z)H'+[\sqrt{4\mu-1}-2i-2a]H=0,$$ which can be taken to standard form by a scaling transformation $z=-2i s$. Invariant solutions are expressed in terms of imaginary arguments.

\end{enumerate}
The subalgebra $M$ does not give an invariant solution. It is obvious that the invariant solutions for the second canonical form can also be expressed in terms of solutions of confluent hypergeometric functions. They can include elementary solutions for the special choices of the parameters $\mu$ figuring in the equation, and $a$ in the subalgebras. Applying Lie point transformations from the full symmetry group will produce more general invariant solutions.

As an application we can consider the radial heat equation
\begin{equation}\label{radial-heat-k}
  u_t=u_{xx}+\frac{k}{x}u_x, \quad k=n-1, \quad k\ne 0,2.
\end{equation}
The transformation $u=x^{-k/2}\tilde{u}$ maps it to the second canonical form with $\mu=-k(k-2)/4$. So all invariant solutions of the initial equation can be obtained from the solutions discussed above. Let us remark that invariant solutions already appeared in \cite{Gungor2002}.  A basis for symmetry is
\begin{equation}\label{basis-radial-heat-k}
  T=\gen t, \quad D=2t\gen t+x\gen x, \quad C=t^2\gen t+xt \gen x-\frac{1}{4}[x^2+2(k+1)t]u\gen u, \quad M=u\gen u.
\end{equation}
In case $k=2$ it is reduced to the heat equation with  a larger symmetry algebra. Indeed, we have simply $K=0$, so $c_2=c_1=c_0=0$. We put $I=x$, $J=-2/x$ in \eqref{6-dim-c20} and find the basis
\begin{equation}\label{basis-radial-heat-k=2}
  \begin{split}
      & v_1=\gen t, \quad v_2=2t\gen t+x\gen x, \quad v_3=t^2\gen t+xt\gen x-\frac{1}{4}(x^2+6t)u\gen u,\\
       & v_4=t\gen x-\left(\frac{x}{2}+\frac{t}{x}\right)u\gen u,\quad
        v_5=\gen x-\frac{1}{x}u\gen u,\quad
        v_6=u\gen u.
   \end{split}
\end{equation}
We refer to \cite{PopovychKunzingerIvanova2008a} for a group classification of $u_t=u_{xx}+b(x)u_x$.

Putting $\mu=-k(k-2)/4$, $a=-(k+2)/4$ and $c_0=(4\pi)^{-(k+1)/2}$ in Eq. \eqref{elem-sol-1} gives the elementary solution (heat kernel) of the radial heat equation \eqref{radial-heat-k} as
\begin{equation}\label{elem-sol-2}
  u=(4\pi t)^{-\frac{1}{2}(k+1)}\exp[-\frac{x^2}{4t}],
\end{equation}
which is nothing else but the fundamental solution $K(x,t,0)$ satisfying $K(x,t,0)=\delta(x)$ as $t\to 0^{+}$. Using an argument introduced by Craddock and Dooly \cite{CraddockDooley2001}, $K(x,t,0)$ can be translated to $K(t,x,y)$ with $\lim_{t\to 0^+ }K(x,t,y)=\delta(x-y)$ by picking an appropriate translation group (at least up to a change of basis). For example, when $k=2$, the vector fields $v_1$ and $v_5$ commute and $v_5$ turns out to be an appropriate element for this purpose. Applying the translation  transformation (also changing $u$) generated by $v_5$ of \eqref{basis-radial-heat-k=2}
  $$\tilde{t}=t  ,\quad \tilde{x}=x-y,  \quad u=\left(1-\frac{y}{x}\right)\tilde{u}(\tilde{x},\tilde{t}),$$ where $y$ is the group parameter, to $\tilde{K}(x,t)\triangleq K(x,t,0)$ we find
  $$K(x,t,y)=\left(1-\frac{y}{x}\right)\tilde{K}(x-y,t)=(4\pi t)^{-3/2}\left(1-\frac{y}{x}\right)\exp[-\frac{(x-y)^2}{4t}].$$
We remark that the  solution \eqref{elem-sol-2} can also be obtained mapping a constant solution $u=c_0$ of the equation by means of nontrivial Lie point symmetries, for example, by   $C$ of \eqref{basis-radial-heat-k}.

\subsection{Heat polynomials}\label{subsection-heat-poly}
Consider the Cauchy problem for the heat equation with the initial data $u(x,0)=x^n$ on $x\in \mathbb{R}$, a homogeneous polynomial of degree $n$. The solutions of this problem on $\mathbb{R}\times (0,\infty)$ can be expressed as a power series in $t$ as
\begin{equation}\label{heat-poly}
  u_n(x,t)=\sum_{j=0}^{\infty}a_j(x)t^j,
\end{equation}
where $a_0(x)=x^n$, $a_j(x)=j^{-1}\partial_x^{2j} x^n$, $j\geq 1$. They are called heat polynomials and can be formally represented  by $u_n=e^{t\partial_x^2}x^n$. They are explicitly expressed by the formula
$$u_n(x,t)=n!\sum_{j=0}^{[\frac{n}{2}]}\frac{x^{n-2j}}{(n-2j)!}\frac{t^j}{j!}.$$ The first five polynomials are given by
$$u_0=1, \quad u_1=x, \quad u_2=x^2+2t, \quad u_3=x^3+6xt, \quad u_4=x^4+12x^2t+12t^2.$$
The heat polynomials are  parabolically-homogenous of degree $n$ in the sense
\begin{equation}\label{homogeneity}
  u_n(\lambda x,\lambda^2 t)=\lambda^n u_n(x,t)
\end{equation}
for all $\lambda>0$.

The heat polynomials can also be generated from a result relating solutions of the heat equation $u_t=u_{xx}$ which was observed in \cite{FushchychShtelenSerovPopovych1992} in the context of Q-conditional symmetry of the heat equation. If $f=f(x,t)$ is a solution of the heat equation  then $Q=f\gen t-f_x \gen x$ is a Q-conditional symmetry of the equation. This leads to the fact that the solution $f$ is related to another solution $u(x,t)$ obtained from integrating the exact equation $f_x dt+f dx=0$ in the form $u(x,t)=C$, a constant. The simple solution $f=1$ generates the heat polynomials with the slightly different initial condition $u(x,0)=x^n/n!$ via a recursive process. A study of heat polynomials from the Lie point symmetry point of view was presented in \cite{Leach2006}.

The heat polynomials are closely related to the Hermite polynomials by $H_n(x)=u_n(x,-1)$ (an Appell sequence). They can also  be recovered in the symmetry context. The above Cauchy problem is left invariant by the dilation generator $D_n=x \gen x+2t\gen t+nu\gen u$, $n\in \mathbb{N}_0$, in other words scale-invariant solutions as being solutions of the PDE
$$x \frac{\partial u_n}{\partial x}+2t\frac{\partial u_n}{\partial t}=nu_n,$$ in other words the solutions of the functional equation \eqref{homogeneity} should produce the heat polynomials. They are of the form
$$u_n=t^{n/2}F(\eta), \quad \eta=\frac{x}{\sqrt{4t}}.$$ When  substituted into the backward heat equation $u_t+u_{xx}=0$, $F$ satisfies a Hermite polynomial equation
$$F''-2\eta F'+2n F=0.$$  We can switch to the forward heat equation by the time reversal $t\to -t$ and obtain the heat polynomials in terms of Hermite polynomials
$$u_n=(-t)^{n/2}{H_n(\eta}),  \quad \eta=\frac{x}{\sqrt{-4t}}.$$
An alternative way to defining heat polynomials is done through the solution
$$U_{z}(x,t)=e^{zx+z^2t}=e^{t\partial_x^2}e^{zx},  \quad U_{z}(x,0)=e^{zx}$$ obtained from the Galilei action $e^{-2z B}$ on the constant solution $u_0=1$, where $z$ is the group parameter. The heat polynomials are also defined by the coefficients of $z^n/n!$ in the expansion of
$$U_{z}(x,t)=\sum_{n=0}^{\infty}\frac{z^n}{n!}u_n(x,t).$$
With this definition, the connection with the Hermite polynomials could also be revealed by picking $\xi=z\sqrt{-t}$, $y=x/(2\sqrt{-t})$ in the generating function formula for the Hermite polynomials
$$e^{2\xi y-\xi^2}=\sum_{n=0}^{\infty}\frac{\xi^n}{n!}H_n(y).$$
The integral representation of $u_n$ is given by
$$u_n(x,t)=\int_\mathbb{R}K(x-y,t)y^ndy.$$
Using the Appell transformation \eqref{Appell-heat} we can produce another set of solutions to the heat equation (called the set associated with the set of heat polynomials \cite{RosenbloomWidder1959})
$$v_n(x,t)=K(x,t)u_n(\frac{x}{t},-\frac{1}{t}),  \quad n\in \mathbb{N}_0,\quad t\in(0,\infty),$$ where $K(x,t)$ is the heat kernel at the origin. By the homogeneity of $u_n$, namely by \eqref{homogeneity} we have
$$v_n=K(x,t)u_n(x,-t)t^{-n}.$$ This set also can be defined as the coefficient of $z^n/n!$ in the expansion
$$K(x-2z,t)=\sum_{n=0}^{\infty}\frac{z^n}{n!}v_n(x,t).$$ There is an analogue formula for $v_n$ in terms of Hermite polynomials
$$v_n=t^{-n/2}K(x,t)H_n(\eta),  \quad \eta=\frac{x}{\sqrt{4t}}.$$
We recall that the sets $\curl{K(x,t)u_n(x,-t)}_{n=0}^{\infty}$  and $\curl{v_n(x,t)}_{n=0}^{\infty}$ are complete in $L(\mathbb{R})$ or $L^2(\mathbb{R})$ for $t>0$.

Another interesting aspect of these two sets of functions is that they are biorthogonal on $\mathbb{R}$
$$\int_{\mathbb{R}}u_m(x,-t)v_n(x,t)dx=2^n n! \delta_{mn},$$ which readily follows from the orthogonality property of the Hermite polynomials
$$\int_{\mathbb{R}} e^{-x^2}H_m(x)H_n(x)dx=c_n \delta_{mn}, \quad c_n=2^n n! \sqrt{\pi}.$$

For $n$-dimensional version of the heat equation, we can think of generalized heat polynomials as quasi-homogeneous polynomials of degree $k$  defined by
$$ u_k(x,t)=e^{t \Delta} P(x)=\sum_{j=0}^{[\frac{k}{2}]}\frac{\Delta^j P(x)}{j!}t^j,  \quad (x,t)\in \mathbb{R}^n\times (0,\infty),  $$
as $C^2$ solutions to the heat equation with the initial condition $u(x,0)=P(x)$, where $P$ is a homogeneous polynomial of degree $k$. The classical multidimensional heat polynomials are obtained by restricting the initial data to the monomial $P(x)=x_1^{k_1}x_2^{k_2}\ldots x_n^{k_n}$, $k_1+k_2\ldots +k_n=k$.

\subsection{Fundamental solutions and applications to initial-boundary value problems}
As we have already encountered before, Lie group theory combined with equivalence transformations can be effectively used to solve boundary-value problems, in particular initial-value  problems like constructing heat kernels (also known as Gaussian kernel, fundamental or source solution, propagator of the diffusion,  or diffusion kernel and even Green's function) for general parabolic (evolution) type equations. We have seen examples of group-invariant solutions which are also fundamental solutions. Craddock and his coworkers (for example, see \cite{CraddockDooley2001, CraddockPlaten2004, CraddockLennox2007, Craddock2009} for scalar parabolic equations  and \cite{CraddockLennox2012, KangQu2012} for parabolic systems) have developed new techniques as an ingenious synthesis of Lie point symmetries and the theory of classical integral transforms (Laplace, Fourier, Mellin, Hankel and others) and have successfully applied them to a number of problems.
Here we intend to discuss some basic ideas on applying symmetry group methods. A distribution $K$ is called a fundamental solution of a linear PDE
if it solves the associated Cauchy problem with the initial condition $K(x,0,y)=\delta(x-y)$, where $\delta$ is the Dirac distribution. If $K(x,t,y)$ is a heat kernel, then the solution of
the Cauchy problem  for the parabolic equations
\begin{equation}\label{TILPE}
\begin{split}
    u_t &=Hu=a(x)u_{xx}+b(x)u_x+c(x)u, \quad a\ne 0,\quad  \quad x\in \Omega \subseteq\mathbb{R},\quad t>0\\
 u(x,0) &=\phi(x)
\end{split}
\end{equation}
is given by the integral formula
\begin{equation}\label{intrep}
  u(x,t)=\int_{\Omega}K(x,t,y)\phi(y)dy,
\end{equation}
provided that the integral converges. For constant coefficient equations the convolution integral
\begin{equation}\label{convol}
  u(x,t)=\int_{\Omega}K(x-y,t,0)\phi(y)dy
\end{equation}
converges to a solution of \eqref{TILPE}, for example if $\phi\in L^2(\Omega)$.

For a backward version of equation \eqref{TILPE} where $t\in [0,T)$ we replace the initial condition for the fundamental solution by a terminal condition $K(x,T,y)=\delta(x-y)$. The transition from the backward to forward form is made possible by the simple change of variable $s=T-t$.
\paragraph{Invariance of boundary and integral conditions:}
One method to solve the Cauchy problem \eqref{TILPE} is to find a group-invariant solution using a subgroup of the full symmetry group that also leave invariant  the boundary (or initial) conditions. Recall that a general element of the nontrivial symmetry algebra of \eqref{TILPE} can be represented by
\begin{equation}\label{gen}
  \mathbf{v}=\tau(t)\gen t+\xi(x,t)\gen x+\phi(x,t)u\gen u,
\end{equation}
where
$\xi$, $\phi$ are defined by \eqref{xi}, \eqref{phi} and $\tau$, $\rho$, $\sigma$ are solutions of either \eqref{splitcomp4} or \eqref{splitcomp6}. So $\mathbf{v}$ can linearly depend on 4 or 6 arbitrary constants
\begin{equation}\label{gen-element}
  \mathbf{v}=\sum_{i=1}^n c_i v_i,  \quad n\in\curl{4,6}.
\end{equation}
\begin{exa}
Invariance of the initial condition $u(x,0)=\delta(x-x_0)$, $x\in \mathbb{R}$, where $\delta$ is the Dirac measure weighted at $x_0$.

First of all, invariance of the boundaries $t=0$ and $x=x_0$   implies that infinitesimally we should have $\mathbf{v}(t)=0$ when $t=0$ and $\mathbf{v}(x-x_0)=0$ when $x=x_0$, which, in terms of the coefficients of   $\mathbf{v}$, are
\begin{equation}\label{boundary-tau-xi}
  \tau(0)=0,  \quad \xi(x_0,0)=0.
\end{equation}
Invariance of the initial condition $u(x,0)=\delta(x-x_0)$ should be interpreted in the distribution sense. We require the relation
$$ \int_{\mathbb{R}}u(x,0)\varphi(x)dx=\varphi(x_0),$$ for every test function $\varphi\in \mathcal{D}(\mathbb{R})=C_0^{\infty}(\mathbb{R})$, to be preserved by the group action of $\mathbf{v}$.  Infinitesimally, this amounts to
$$\left[\mathbf{v}(F)+F\xi_x\right]\Big|_{x=x_0, t=0}=0,  \quad F=u(x,t)\varphi(x).$$ From this and the second condition of  \eqref{boundary-tau-xi}  a further condition
\begin{equation}\label{boundary-phi}
  \phi(x_0,0)=-\xi_x(x_0,0)
\end{equation}
follows. The subgroups satisfying the above three conditions \eqref{boundary-tau-xi}-\eqref{boundary-phi} can be applied to obtain  special group-invariant solutions which are supposed to produce heat kernels. In the special case when the symmetry algebra is six-dimensional ($n=6$ in \eqref{gen-element}), these conditions impose the following
\begin{equation}\label{sub-conds}
  \tau(0)=0, \quad \rho(0)=-\frac{1}{2}\dot{\tau}(0)I_0,  \quad \sigma(0)=\frac{1}{8}\ddot{\tau}(0)I_0^2+\frac{1}{2}\dot{\rho}(0)I_0-\frac{1}{2}\dot{\tau}(0),
\end{equation}
where $I_0\triangleq I(x_0)$. This means that the subalgebra that will produce heat kernel can be at most three dimensional. If the symmetry algebra is four-dimensional, then the above conditions boil down to
\begin{equation}\label{sub-conds-2}
  \tau(0)=\dot{\tau}(0)=0,  \quad \sigma(0)=\frac{1}{8}\ddot{\tau}(0)I_0^2.
\end{equation}
The corresponding subalgebra is one-dimensional. For a class of equations in potential form in two space dimensions, a systematic study of fundamentals solutions based on symmetry appeared in \cite{LaurenceWang2005}.  Very recently, the same procedure has been applied to a 2-dimensional ultra-parabolic Fokker-Planck-Kolmogorov equation in \cite{KovalenkoStogniyTertychnyi2014}.
\end{exa}
\begin{exa}The subalgebra $\lr{J_{kl},B_k}$ of \eqref{Lie-n} leaves the
$n$-dimensional heat equation and the condition  $u(x,0)=\delta(x)$, $x\in \mathbb{R}^{n}$. The solution invariant under this subalgebra should have the form
$$u(x,t)=e^{-\frac{|x|^2}{4t}}F(t),  \quad t>0.$$ On substituting into the equation we find the reduced ODE as
$2tF'+nF=0$ with the general solution $F(t)=c_0t^{-n/2}$. The requirement $\int_{\mathbb{R}^{n}} u(x,t)dx=1$ or the limit \eqref{distrib-limit} leads to the heat kernel $u=K(x,t)=(4\pi t)^{-n/2}\exp(-|x|^2/(4t))$ with source at the origin.

Similarly, let us require invariance under the subalgebra $\lr{J_{kl}, C}$ of \eqref{Lie-n}. Invariants of the rotation group are $|x|$, $t$ and $u$. The invariant solution is found by solving the characteristic equation associated to $C$. It is given by
$$u=t^{-n/2}e^{-\frac{|x|^2}{4t}}F(\eta),  \quad \eta=\frac{|x|}{t},$$ where $F$ is a linear function in $\eta$. Again, $K(x,t)$ is recovered by the special choice of the arbitrary constants in $F$. We leave it to the reader to see how it can also be obtained by using the subalgebra $ \lr{J_{kl}, D+aM}$.
\end{exa}

On the other hand, given that the equation is equivalent to the heat equation, equivalence group can be used to obtain the fundamental solution from \eqref{fundsol}.
An alternative method for heat kernels  was introduced by  Craddock and his collaborators. The idea is simply to relate a  nontrivial  solution obtained from  a stationary solution by symmetry transformation to the integral representation \eqref{intrep} of the solution. We give here a brief description of the idea. Let the action of a symmetry vector field $\mathbf{v}$ on solutions be given in the form
$$\tilde{u}(x,t)=e^{\varepsilon \mathbf{v}}u(x,t)=\mu(x,t,\varepsilon)u(a_1(x,t,\varepsilon),a_2(x,t,\varepsilon))$$   for some known functions $\mu$, $a_1$ and $a_2$. If a stationary solution $u_0(x)$ is applied to it, we get $$U_{\varepsilon}(x,t)=e^{\varepsilon \mathbf{v}}u_0(x)=\mu(x,t,\varepsilon)u_0(a_1(x,t,\varepsilon)).$$ We require the solution $U_{\varepsilon}(x,t)$ to satisfy
\begin{equation}\label{int-eq}
  U_{\varepsilon}(x,t)=\int_{\Omega}U_{\varepsilon}(y,0)K(x,t,y)dy, \quad \Omega\subseteq\mathbb{R}.
\end{equation}
Therefore, heat kernels then arise as a standard integral transform of the  solution $U_{\varepsilon}(x,t)$ with the integral equation kernel $U_{\varepsilon}(y,0)$. The group parameter $\varepsilon$  plays the role of the integral transform parameter.  Different choices of $u_0$ in general lead to different heat kernels. Observe that if $u_0=1$, then from \eqref{int-eq} it follows that $\int_{\Omega}K(x,t,y)dy=1$ since $U_{0}(x,t)=1$.   This method has been applied to many interesting diffusion processes in a series of papers \cite{CraddockDooley2001, CraddockPlaten2004, CraddockLennox2007, Craddock2009, CraddockKonstandatosLennox2009, CraddockLennox2009, BaldeauxPlaten2013}. Theoretical basis of this method and other related ones are found in these works.

We would like to conclude this subsection by driving heat kernels for some equations and compare the above-mentioned methods.
\begin{exa}
Calculation of the fundamental solution of the second canonical form \eqref{2ndcanonical} with $x\geq 0$.
\begin{itemize}
  \item Fundamental solution as group-invariant solution:
   Under the conditions \eqref{sub-conds-2} the original symmetry algebra \eqref{canonical-2-basis} is reduced to the one-dimensional subalgebra generated by
      $$\mathbf{v}=t^2\gen t+xt\gen x-\frac{1}{4}(x^2-y^2+2t)u\gen u.$$ The group invariant-solution should have the form
      $$u=t^{-1/2}e^{-\frac{1}{4t}(x^2+y^2)}F(z),  \quad z=\frac{x}{t},$$ and $F$ satisfies the ODE
      $$z^2F''+(\mu-\frac{y^2}{4}z^2)F=0.$$ The solution of the ODE is expressed in terms of the modified Bessel function of the first kind (the other independent  solution  is discarded  because the  Bessel function $I_{-\nu}$ is not integrable near zero for $\nu\geq 1$)
      $$F=c_0(y) z^{1/2}I_{\nu}\left(\frac{yz}{2}\right), \quad \nu=\frac{\sqrt{1-4\mu}}{2}.$$
      We replace $I_{\nu}$ by $K_{\nu}$ if $\nu$ is an integer.
      Finally we recover the fundamental solution
      $$K(x,t,y)=c_0\frac{\sqrt{x}}{t}e^{-\frac{1}{4t}(x^2+y^2)}I_{\nu}\left(\frac{xy}{2t}\right), \quad \nu=\frac{\sqrt{1-4\mu}}{2}$$
      up to the normalization constant $\displaystyle c_0(y)=\frac{\sqrt{y}}{2}$ that will come from the condition
      $\lim_{t\to 0^{+}}\int_{0}^{\infty}K(x,t,y)dx=1$.

      \item
We shall use the projective symmetry
$$C=t^2\gen t+xt\gen x-\frac{1}{4}(x^2+2t)u\gen u,$$
to construct the fundamental solution from the stationary (time independent) solution $$u_0(x)=x^{\varrho_1}, \quad \varrho_1=\frac{1}{2}(1-\sqrt{1-4\mu}).$$ Observe that $u_0(x)=1$ when $\mu\to 0$. The other stationary solution $u_0(x)=x^{\varrho_2}$ for $\varrho_2=\frac{1}{2}(1+\sqrt{1-4\mu})$ becomes nonconstant when $\mu\to 0$. This solution is disposed of because the fundamental solution that  comes from this will not reduce to the necessary transition probability density with the property
$\int_0^{\infty}K(x,t,y)dy=1$.

Exponentiating $C$ gives the solution transformation formula
\begin{equation}
  \begin{split}
       \tilde{t}&=\frac{t}{1+\lambda t},  \quad \tilde{x}=\frac{x}{1+\lambda t}, \\
        \tilde{u}(x,t)&=(1+\lambda t)^{-1/2}\exp\curl{\frac{-\lambda x^2}{4(1+\lambda t)}} u(\tilde{x},\tilde{t}),
   \end{split}
\end{equation}
where $\lambda$ is the group parameter.
The stationary solution is mapped to the characteristic solution
$$U_{\lambda}(x,t)=(1+\lambda t)^{-1/2-\varrho_1} x^{\varrho_1}\exp\curl{\frac{-\lambda x^2}{4(1+\lambda t)}}$$ with $U_{\lambda}(x,0)=x^{\varrho_1}\exp\curl{-\lambda x^2/4}$. We rewrite it in the form
$$U_{\lambda}(x,t)=t^{-\alpha}(\lambda+t^{-1})^{-\alpha}x^{\varrho_1}\exp\curl{-\frac{x^2}{4t}}
\exp\curl{\frac{x^2}{4t^2(\lambda+t^{-1})}},$$
where $\alpha=1-\frac{\sqrt{1-4\mu}}{2}$.
We substitute it to the integral equation \eqref{int-eq}
$$\int_0^{\infty}y^{\rho_1}\exp[-\frac{\lambda}{4}y^2]K(x,t,y)dy=U_{\lambda}(x,t),$$ which is converted to
$$\int_0^{\infty}e^{-\lambda z}\hat{K}(x,t,2\sqrt{z})dz=U_{\lambda}(x,t)$$ by the substitution $z=y^2/4$. Here
$\hat{K}(x,t,2\sqrt{z})=2^{\rho_1}z^{(\rho_1-1)/2}K(x,t,2\sqrt{z})$.
$K$ is recovered by inverting $U_{\lambda}(x,t)$ from $\lambda$ to $z=y^2/4$
\begin{equation}\label{fundsol-can}
  K(x,t,y)=\left(\frac{\sqrt{xy}}{2t}\right)\exp\curl{-\frac{x^2+y^2}{4t}}I_{\frac{\sqrt{1-4\mu}}{2}}
\left(\frac{xy}{2t}\right)
\end{equation}
using the Laplace transform inversion formula from $\lambda$ to $z$ \cite{Sneddon1972}
\begin{equation}\label{Lap-tr}
  \mathscr{L}^{-1}\curl{(\lambda+\beta)^{-\alpha}\exp\left(\frac{a}{\lambda+\beta}\right)}=e^{-\beta z}\left(\frac{z}{a}\right)^{(\alpha-1)/2}I_{\alpha-1}(2\sqrt{a z}), \quad \alpha>0,
\end{equation}
where $I_{\alpha-1}$ denotes the modified Bessel function of  the first kind of order $\alpha-1$.
The fundamental solution becomes elementary when $\mu=-N/2$, $N\in \mathbb{Z}$. In the limit $\mu\to 0$ the fundamental solution of the heat equation is recovered when the linear combination $I_{1/2}(y)+I_{- 1/2}(y)=\sqrt{2/(\pi y)}e^y$ of the elementary functions $I_{\pm 1/2}$ is taken into account.

We can transform \eqref{fundsol-can} by ${K}_n(x,t,y)=(x/y)^{(1-n)/2}K(x,t,y)$, $n\ne 1,3$ to obtain the fundamental solution of the $n$-dimensional radial heat equation \eqref{radial-heat-k} in the form
\begin{equation}\label{fundsol-radial}
  \tilde{K}_n(x,t,y)=\frac{1}{2t}x^{1-\frac{n}{2}}y^{n/2}e^{-\frac{(x^2+y^2)}{4t}}I_{\frac{n}{2}-1}\left(\frac{xy}{2t}\right).
\end{equation}
\end{itemize}
\end{exa}

\begin{exa}
The forward Black-Scholes equation
\begin{equation}\label{F-Black-Scholes}
 u_t=\frac{1}{2}\sigma^2x^2u_{xx}+rxu_x-ru=0,  \quad \sigma\ne 0.
\end{equation}
We have discussed symmetries of the backward Black-Scholes equation (see Example \eqref{ex-BBS}). A basis is easily obtained   by putting $c_2=c_1=0$, $c_0=-\frac{1}{2\sigma^2}\ell-r$, $\ell=r-\sigma^2/2$, $I=(\sqrt{2}/\sigma)\ln x$, $J=-(\sqrt{2}/\sigma)\ell$ in \eqref{6-dim-c20}:
\begin{equation}\label{basis-FBS}
  \begin{split}
      & v_1=\gen t, \quad v_2=x\gen x, \quad v_3=2t\gen t+(\ln x-\ell t)x\gen x-2rtu\gen u, \\
      &  v_4=-\sigma^2 xt \gen x+(\ln x+\ell t)u\gen u, \\
      &  v_5=2\sigma^2 t^2\gen t+2\sigma^2 xt \ln x\gen x-[(\ln x+\ell t)^2+2\sigma^2rt^2+\sigma^2 t]u\gen u,\quad v_6=u\gen u.
   \end{split}
\end{equation}
Below we present three different methods for the derivation of the heat kernel:
\begin{itemize}
  \item Heat kernel as group invariant solution: The conditions \eqref{sub-conds} for $x_0=y$ reduce the algebra to the subalgebra spanned by
\begin{equation}\label{subalg}
  \begin{split}
      & X_1=2t\gen t+[-\ell t+\ln \frac{x}{y}]x\gen x+(2rt-1)u\gen u, \\
       & X_2=-\sigma^2 xt \gen x+[\ell t+\ln \frac{x}{y}]u\gen u, \\
       & X_3=2\sigma^2 t^2\gen t+2\sigma^2 xt\ln x\gen x-[(\ln x+\ell t)^2+2r\sigma^2t^2+\sigma^2 t-(\ln y)^2]u\gen u.
   \end{split}
\end{equation}
Solution invariant under the subalgebra generated by the subalgebra  $\lr{X_1, X_2, X_3}$ should have the form
\begin{equation}\label{group-inv-kernel}
  \begin{split}
      & K(x,t,y)=c_0(y)t^{-1/2}x^{-A(t,y)}e^{-B(x,t,y)}, \quad t>0, \\
       & A=\frac{\ell t-\ln y}{\sigma^2 t}, \quad B=\frac{(\ln x)^2+(\ln y)^2}{2\sigma^2 t}+\left(\frac{\ell^2}{2\sigma^2}+r\right)t,
   \end{split}
\end{equation}
where  $c_0(y)$ is to be determined from the initial condition $u(x,0)=\delta(x-y)$.
This is the heat kernel up to a nonzero multiplicative constant. In \cite{GazizovIbragimov1998}, $c_0$ was computed to be
$$c_0(y)=\frac{1}{\sigma y\sqrt{2\pi}}\exp[\frac{\ell}{\sigma^2}\ln y]$$ by means of some manipulations of distributional limits.
  \item For the sake of completeness we include the integral transform technique applied to \eqref{F-Black-Scholes} in Ref. \cite{CraddockPlaten2004}. The vector field  $v_4$ exponentiates to give the group transformation of the solution $u$
      $$\tilde{u}_{\varepsilon}(x,t)=x^{-\varepsilon}\exp\left[(\frac{\sigma^2}{2}\varepsilon^2-\mu \varepsilon)t\right]u(xe^{-\sigma^2\varepsilon t},t),$$ where $\varepsilon$ is the group parameter. We apply it to the stationary solution  $u_0=x$ and obtain a new solution
      $$U_{\varepsilon}(x,t)=\tilde{u}_{\varepsilon}=x^{1-\varepsilon}
      \exp{\curl{\frac{\varepsilon}{2}[(\varepsilon-1)\sigma^2-2r] t}},\quad U_{\varepsilon}(y,0)=y^{1-\varepsilon}.$$  Substituting $U_{\varepsilon}(x,t)$ into the integral equation \eqref{int-eq} and performing the change of parameter $\varepsilon=2-s$ gives
      $$\exp\curl{{\frac{1}{2}(s-2)[(s-1)\sigma^2+2r]t}}x^{s-1}=\int_{0}^{\infty}y^{s-1}K(x,t,y)dy,$$
     which is recognized as a Mellin integral equation.  $K$ is obtained  by an inverse Mellin transform of the left side from $s$ to $y$
     \begin{equation}\label{BS-kernel}
       K(x,t,y)=\frac{e^{-rt}}{\sigma y\sqrt{2\pi t}}\exp\curl{-\frac{[\ln \frac{x}{y}+\ell t]^2}{2\sigma^2 t}},
     \end{equation}
      which coincides with \eqref{group-inv-kernel}. The inversion can be performed using the connection of the Mellin transform with the Fourier transform
     $$\mathscr{M}\curl{f(x)}(s)=\sqrt{2\pi}\mathscr{F}\curl{f(e^{-x})}(is).$$ We have taken the Fourier transform $\mathscr{F}(f(x))=\frac{1}{\sqrt{2\pi}}\int_{-\infty}^{\infty}f(x)e^{iyx}dy$.
  \item The last method is to transform the heat kernel of the standard heat equation by means of the Black-Scoles transformation \eqref{BS-transform} discussed before for the backward Black-Scholes equation. The standard heat kernel
      $$\tilde{K}(\tilde{x},t,0)=\frac{1}{\sqrt{4\pi t}}\exp[-\frac{\tilde{x}^2}{4t}]$$ transforms into \eqref{BS-kernel} by the transformation \eqref{BS-transform} with the choice  $\omega_1=0$, $\omega_0=-\frac{\sqrt{2}}{\sigma}\ln y$, $\mathcal{D}=\ell$ and
      $$\tilde{u}(\tilde{x},\tilde{t})=\tilde{K}(\tilde{x},\tilde{t},0)=\tilde{K}(\frac{\sqrt{2}}{\sigma}\ln\frac{x}{y},t,0), \quad \nu_0=\sqrt{2}(y\sigma)^{-1}y^{l/\sigma^2}.$$
\end{itemize}
\end{exa}

\begin{exa}
The Ornstein-Uhlenbeck process (Fokker-Planck version).
\begin{equation}\label{OU-eq}
  u_t=\frac{\sigma^2}{2}u_{xx}+(bxu)_x, \quad b>0, \quad \sigma>0.
\end{equation}
From example \eqref{ex-arb-drift} we already know that this equation is equivalent to the heat equation. Its symmetry algebra is obtained by taking $c_2=-b^2/4$, ($\kappa=b/2$) $c_1=0$, $c_0=b/2$ in \eqref{6-dim-c2n}
\begin{equation}\label{basis-FBS-2}
  \begin{split}
      & v_1=\gen t, \quad v_2=e^{2bt}\gen t+bxe^{2bt}\gen x-\frac{2b^2}{\sigma^2}e^{2bt}x^2u\gen u,\\
      & v_3=e^{-2bt}\gen t-be^{-2bt}x\gen x+be^{-2bt}u\gen u, \quad v_4=\sigma e^{bt}\gen x-\frac{2b}{\sigma}e^{bt}xu\gen u,\\
      & v_5=e^{-bt}\gen x,  \quad v_6=u\gen u.
\end{split}
\end{equation}
The well-known heat kernel of Eq. \eqref{OU-eq} is constructed by Fourier transform technique. Here we apply Lie symmetry methods.
\begin{itemize}
  \item Heat kernel as group invariant solution: We use the subalgebra obtained from the conditions \eqref{sub-conds}
\begin{equation}\label{subalg-2}
  \begin{split}
      & X_1=(1-e^{-2bt})\gen t+be^{-2bt}(x-e^{bt}y)\gen x-be^{-2bt} u\gen u, \\
       & X_2=(e^{2bt}-e^{-2bt})\gen t+be^{-2bt}[(1+e^{4bt})x-2e^{-2bt}y]\gen x+\\
       & -b[1+e^{-2bt}+\frac{2b}{\sigma^2}e^{2bt}x^2-\frac{2b}{\sigma^2}y^2]u\gen u, \\
       & X_3=(e^{bt}-e^{-bt})\gen x+\frac{2b}{\sigma^2}(e^{bt}x-y)u\gen u.
   \end{split}
\end{equation}
There is a single invariant $I(x,t,y)$ of the subalgebra $\lr{X_1,X_2,X_3}$ ($n-r=3-2=1$). It is found from solving the first order system of PDEs $X_1(I)=0$, $X_2(I)=0$, $X_3(I)=0$. From the last equation we find $I=F(t,\omega)$, where
$$\omega=u\exp\left[\frac{b}{\sigma^2}\frac{(x-ye^{-bt})^2}{1-e^{-2bt}}\right].$$ The first equation gives $I=F(\zeta)$, $\zeta=\omega(1-e^{-2bt})^{1/2}$ and the second one is automatically  satisfied. The invariant solution should have the form
$$u=M(1-e^{-2bt})^{-1/2}\exp\left[-\frac{b}{\sigma^2}\frac{(x-ye^{-bt})^2}{1-e^{-2bt}}\right],$$  where $M$ is a constant to be determined from the condition $\int u(x,t,y)dx=1$ and is given by $M=\displaystyle (\frac{b}{\pi \sigma^2})^{1/2}$. The corresponding solution is the fundamental solution.
\item Eq. \eqref{OU-eq} admits the Gaussian distribution solution with mean zero and variance $\sigma^2/2b$
$$u_0(x)=\sqrt{\frac{b}{\pi \sigma^2}}e^{-\frac{b x^2}{\sigma^2}}.$$   We look at the transformation of  a solution $u(x,t)$ under the group generated by $v_3$. It is given by
$$U_{\varepsilon}(x,t)=(1-2b \varepsilon e^{-2bt})^{-1/2}u\left(\frac{x}{\sqrt{1-2b \varepsilon e^{-2bt}}},\frac{1}{2b}\ln(e^{2bt}-2b \varepsilon)\right).$$ $U_{\varepsilon}(x,t)$ is a solution whenever $u(x,t)$ is. The action  on the stationary solution with the choice $2b \varepsilon=1$ induces the fundamental solution $K(x,t,0)$ at $y=0$
$$K(x,t,0)=\sqrt{\frac{b}{\pi \sigma^2}}(1- e^{-2bt})^{-1/2}\exp\left[-\frac{b}{\sigma^2}\frac{x^2}{1-e^{-2bt}}\right].$$ We can use the translational symmetry $x\to x-e^{bt}y$ generated by $v_5$ to recover the full fundamental solution  $$K(x,t,y)=K(x-e^{bt}y,t,0)=\sqrt{\frac{b}{\pi \sigma^2}}(1- e^{-2bt})^{-1/2}\exp\left[-\frac{b}{\sigma^2}\frac{(x-e^{bt}y)^2}{1-e^{-2bt}}\right],$$
which is the transition probability density (see \cite{Pavliotis2014}). Note that in the limit $t\to \infty$, $K(x,t,y)$ tends to the stationary solution $u_0(x)$ (equilibrium density).
\item Transformation to the heat equation:

On using formula \eqref{ToHeatTr2}, we find the transformation
$$\tilde{t}=\frac{1}{2b}(e^{2bt}-1), \quad \tilde{x}=\frac{\sqrt{2}}{\sigma}(e^{bt}x-y), \quad u=c_0 e^{bt}\tilde{u}(\tilde{x},\tilde{t})$$
mapping \eqref{OU-eq} to the heat equation for $\tilde{u}$.
The standard heat kernel
$$\tilde{u}=\tilde{K}(\tilde{x},\tilde{t},0)=\frac{1}{\sqrt{4\pi \tilde{t}}}e^{-\frac{\tilde{x}^2}{4\tilde{t}}}$$
is transformed to the heat kernel $K(x,t,y)=c_0e^{bt}\tilde{K}(\tilde{x},\tilde{t},0)$ for \eqref{OU-eq}. From the initial condition
$$K(x,0,y)=c_0\tilde{K}(\tilde{x}(x,0),0,0)=c_0\delta\left(\frac{\sqrt{2}}{\sigma}(x-y)\right)=
c_0\frac{\sigma}{\sqrt{2}}\delta(x-y),$$ we find $\displaystyle c_0=\sqrt{\frac{2}{\sigma^2}}$. We have used the property $\delta(\lambda x)=\lambda^{-1}\delta(x)$, $\lambda>0$.
\end{itemize}
\end{exa}

\begin{exa}
Two-dimensional time-dependent heat equation
\begin{equation}\label{vcheat}
  u_t=u_{xx}+t^{-2}u_{yy}, \quad t>0.
\end{equation}
A basis for the symmetry algebra of \eqref{vcheat} is given by
\begin{equation}\label{Xi}
\begin{split}
  &X_1=\gen t-\frac{y}{t}\gen y-(\frac{y^2}{4}-\frac{1}{2t})u\gen u, \\
  &X_2=2t\gen t+x\gen x-y\gen y, \\
  &X_3=t^2\gen t+xt\gen x-(\frac{x^2}{4}+\frac{t}{2})u\gen u,\\
  &X_4=ty \gen x-\frac{x}{t}\gen y-\frac{xy}{2}u\gen u,\\
  &X_5=t\gen x-\frac{x}{2}u\gen u,\\
  &X_6=\frac{1}{t} \gen y+\frac{y}{2}u\gen u,\\
  &X_7=\gen x,\quad X_8=\gen y,  \quad X_9=u\gen u.
\end{split}
\end{equation}
The Lie symmetry algebra $\lie$ is identified as a 9-dimensional algebra with the structure
\begin{equation}\label{symalg}
\lie=  \Sl(2,\mathbb{R}) \oplus \So(2) \vartriangleright \nil(5)\sim \langle X_1, X_2, X_3\rangle \oplus \langle X_4\rangle
\vartriangleright \langle X_5, X_6, X_7, X_8, X_9 \rangle,
\end{equation}
where $\nil(5)$ is the 5-dimensional Heisenberg algebra with center $X_9$.
We see that the Lie symmetry algebra of \eqref{vcheat} is isomorphic to that of the standard heat equation
\begin{equation}\label{standardpsd}
  u_t=u_{xx}+u_{yy}.
\end{equation}
This suggests that there should be a local point transformation relating these two equations.

We shall make use of the symmetries of \eqref{vcheat} to construct the heat kernel   of \eqref{vcheat}. The heat kernel $K(x,y,t\,;x_0,y_0,t_0)$ at the point $(x_0,y_0,t_0)$ is a distribution function in the whole $(x,y)$ plane satisfying the initial condition
\begin{equation}\label{initial}
\lim_{t\to t_0} K(x,y,t\,;x_0,y_0,t_0)=\delta(x-x_0)\delta(y-y_0).
\end{equation}
We look for a solution invariant under the symmetry algebra leaving invariant the above initial condition.
We consider the general element of the symmetry algebra
$$X=\tau\gen t+\xi\gen x+\eta \gen y+\phi u\gen u=\sum_{i=1}^n a_i X_i,$$ where $a_i$ are constants. The invariance requirement of \eqref{initial} amounts to the following four conditions on the coefficients of the infinitesimal symmetry generator
$$\tau(t_0)=0, \quad \xi(x_0, y_0,t_0)=0, \quad \eta(x_0, y_0,t_0)=0$$ and
$$\tau'(t_0)+\xi_x(x_0, y_0,t_0)+\eta_y(x_0, y_0,t_0)+\phi(x_0, y_0,t_0)=0.$$
These conditions will reduce the dimension of the symmetry algebra from nine to five.
Applying them to our equation provides the following relations
\begin{equation}\label{conds}
\begin{split}
& a_1+2 a_2 t_0+a_3 t_0^2=0, \quad a_7 + a_5 t_0 + a_2 x_0 + a_3 t_0 x_0 + a_4 t_0 y_0=0, \\
& a_8 -\frac{ a_6}{t_0} - \frac{a_4 x_0}{t_0} - a_2 y_0 - \frac{a_1 y_0}{t_0}=0,  \\
&  a_9 - \frac{a_1}{2 t_0} + \frac{a_3 t_0}{2} - \frac{1}{2} a_5 x_0 - \frac{1}{4}a_3 x_0^2 - \frac{1}{2} a_6 y_0 - \frac{1}{2} a_4 x_0 y_0 -\frac{1}{4} a_1 y_0^2=0.
\end{split}
\end{equation}
Solving this system for the coefficients $\{a_2, a_7, a_8, a_9\}$ in terms of the remaining coefficients  $\{a_1, a_3, a_4, a_5, a_6\}$ and substituting in $X$ we find a 5-dimensional subalgebra spanned by the operators
\begin{eqnarray*}
  Y_1&=&-(t-t_0)\gen t-\frac{x-x_0}{2}\gen x+ \frac{t(y+y_0)-2t_0y}{2t} \gen y+\frac{1}{2}\left(\frac{t_0}{t}+1-2t_0(y^2-y_0^2)\right)u\gen u, \\
   Y_2&=& t(t-t_0)\gen t+\left[xt-\frac{1}{2}(x+x_0)t_0\right]\gen x+\frac{t_0(y-y_0)}{2}\gen y-\frac{1}{2}\left[(t+t_0)+\frac{1}{2}(x^2+x_0^2)\right]u\gen u,\\
   Y_3&=&(ty-t_0 y_0)\gen x-\left(\frac{x}{t}-\frac{x_0}{t_0}\right)\gen y-\frac{1}{2}(xy-x_0y_0)u\gen u,  \\
   Y_4&=&(t-t_0)\gen t-\frac{1}{2}(x-x_0)u\gen u, \\
   Y_5&=&  \left(-\frac{1}{t}+\frac{1}{t_0}\right)\gen y-\frac{1}{2}(y-y_0)u\gen u.
\end{eqnarray*}

Invariants of the subalgebra  $\langle Y_1, Y_2, Y_3, Y_4, Y_5\rangle$ are found by solving the system of first order PDEs
$$Y_i I(x,y,t,u)=0,  \quad i=1,2,3,4,5.$$
We find that there is a single invariant
$$
I=\frac{t-t_0}{\sqrt{t}}\exp\left[\frac{(x-x_0)^2}{4(t-t_0)}+\frac{tt_0(y-y_0)^2}{4(t-t_0)}\right]u.$$
The invariant heat kernel will be obtained from $I=C$, where $C$ is a constant yet to be determined.
This gives the invariant solution
\begin{equation}
\tilde{K}(x,y,t\,;x_0,y_0,t_0)=C \frac{\sqrt{t}}{t-t_0}\exp\Bigl[-\frac{(x-x_0)^2}{t-t_0}-\frac{t_0 t (y-y_0)^2}{t-t_0}\Bigr].
\end{equation}
From the initial condition \eqref{initial} it follows that $C=\sqrt{t_0}/(4\pi)$. This implies that we have found the
kernel of  \eqref{vcheat}
\begin{equation}\label{heatkernel}
 K(x,y,t\,;x_0,y_0,t_0)=\frac{1}{4\pi} \frac{\sqrt{t_0t}}{t-t_0}\exp\Bigl[-\frac{(x-x_0)^2}{4(t-t_0)}-
\frac{t_0 t (y-y_0)^2}{4(t-t_0)}\Bigr]
\end{equation}
with the property
\begin{equation}
\iint_{\mathbb{R}^2}K(x,y,t\,;x_0,y_0,t_0) dxdy=1.
\end{equation}

Another formulation of the kernel  can be obtained using the transformation
$$u= \sqrt{t}e^{\frac{ty^2}{4}}v(x,z,t), \quad z=ty$$ taking \eqref{vcheat} to $v_t=v_{xx}+u_{zz}$, which is known to have the heat kernel
\begin{equation}
K_0(x,z,t\,;x_0,z_0,t_0)=\frac{1}{4\pi (t-t_0)}\exp\left[-\frac{(x-x_0)^2+(z-z_0)^2}{4(t-t_0)}\right].
\end{equation}
Hence the  kernel  will have the form
$$K=K_0\frac{\sqrt{t}}{4\pi(t-t_0)}e^{\frac{ty^2}{4}}\exp\left[-\frac{(x-x_0)^2+(z-z_0)^2}{4(t-t_0)}\right].$$
Now observing the relation
$$\exp[\frac{t y^2}{4} ]\exp[\frac{-(yt-y_0t_0)^2}{4(t-t_0)}]=\exp[\frac{t_0y_0^2}{4}]\exp\left[-\frac{t_0t(y-y_0)^2}{4(t-t_0)}\right]$$
and choosing the nonzero constant $K_0=\sqrt{t_0}\exp[-(t_0y_0^2)/4]$ we recover the heat kernel \eqref{heatkernel}.

\end{exa}

\section{Summary}\label{Section5}
The ubiquitous  linear parabolic partial differential equations of the form \eqref{LPE} are very significant both form mathematical  and physical point of view. There exists an enormous amount of literature devoted to their applications and solution methods. A brief overview of the existing literature is given from  equivalence and symmetry standpoint.
The main motivation of the present paper is to give a unified formulation of transformation and symmetry group properties of this general class.  Two  issues have been the main focus of this paper. One is to establish criteria for the equations under study to be transformable to one of the two canonical forms for which nontrivial symmetries (four or six dimensional other than superposition principle) exist and in particular to give a general transformation formula in case when they are transformable to the heat equation.  The other is  related to the first one: to know when the equations possess nontrivial Lie symmetry  algebras. Some attempts towards answering these questions can be found in the literature. We reconsider these issues in a new approach.  Two criteria based on the knowledge of invariant of the given class of equations are proposed  and applied to several examples.
Lie symmetry properties of all known special cases  that already appeared in the literature  can be immediately recovered by our approach. Of course, this is the case for any other equation  within the class \eqref{mainLPE}.  As part of applications of our results we also discuss methods for constructing fundamental solutions and illustrate with examples. The equivalence group is put to good use to derive the heat kernels for heat equations with linear and quadratic potential (Mehler's formula) in 1+1 and higher dimensions as well.

\subsection*{Acknowledgments}The author thanks P. Winternitz for useful discussions.
%

\end{document}